\documentclass{article}
\usepackage{amsmath}       
\usepackage{amsthm}        
\usepackage{amssymb}       
\usepackage{braket}       
\usepackage{amsfonts}
\usepackage[all]{xy}
\usepackage{xspace}
\usepackage{calc}
\usepackage{comment}
\usepackage{pgfplots}
\usepgflibrary{fpu}
\usepackage{mathtools}
\usepackage{tabularx}
\usepackage{ifthen}             
\usepackage{graphicx}
\usepackage{docmute}
\usepackage{array}
\usepackage{subfloat} 
\usepackage{textcomp}
\usepackage{bbm}            
\usepackage{etoolbox}  
\usepackage{verbatim}
\usepackage{stmaryrd}
\usepackage{cancel}
\usepackage{xcolor}
\usepackage{tensor}
\usepackage{enumerate}  
\usepackage{thm-restate}
\usepackage{enumitem}
\usepackage[export]{adjustbox}

\usepackage[colorlinks=true, linkcolor=black, citecolor=black, urlcolor=black, hypertexnames=false
        ]{hyperref}
\usepackage{dcolumn}
\usepackage{bm}
\usepackage{authblk}
\usepackage[margin=3cm]{geometry}

\DeclareRobustCommand{\SkipTocEntry}[5]{}

\usepackage{tikz}
\usetikzlibrary{matrix}
\usetikzlibrary{decorations.pathreplacing}
\usetikzlibrary{decorations.markings}
\usetikzlibrary{arrows}
\usetikzlibrary{calc}
\usetikzlibrary{shapes.misc}
\usetikzlibrary{fit}
\usepgflibrary{decorations.pathmorphing}
\usepgflibrary{shapes.geometric}

\pgfdeclarelayer{edgelayer}
\pgfdeclarelayer{nodelayer}
\pgfdeclarelayer{foreground}
\pgfsetlayers{edgelayer,nodelayer,main,foreground}

\tikzstyle{none}=[inner sep=0pt]

\tikzstyle{rn}=[circle,fill=Red,draw=Black,line width=0.8 pt]
\tikzstyle{gn}=[circle,fill=Lime,draw=Black,line width=0.8 pt]
\tikzstyle{bl}=[circle,fill=Blue,draw=Black,line width=0.8 pt]

\tikzstyle{simple}=[-,draw=Black,thick]
\tikzstyle{arrow}=[-,draw=Black,postaction={decorate},decoration={markings,mark=at position .5 with {\arrow{>}}},thick]
\tikzstyle{tick}=[-,draw=Black,postaction={decorate},decoration={markings,mark=at position .5 with {\draw (0,-0.1) -- (0,0.1);}},line width=2.000]
\makeatother\def\thickness{0.7pt}
\tikzstyle{dot}=[circle, draw=black, fill=black, inner sep=.5ex, line width=\thickness, node on layer=foreground]

\makeatletter
\pgfkeys{%
  /tikz/on layer/.code={
    \pgfonlayer{#1}\begingroup
    \aftergroup\endpgfonlayer
    \aftergroup\endgroup
  },
  /tikz/node on layer/.code={
     \gdef\node@@on@layer{%
      \setbox\tikz@tempbox=\hbox\bgroup\pgfonlayer{#1}\unhbox\tikz@tempbox\endpgfonlayer\egroup}
    \aftergroup\node@on@layer
  },
  /tikz/end node on layer/.code={
    \endpgfonlayer\endgroup\endgroup
  }
}

\def\node@on@layer{\aftergroup\node@@on@layer}

\makeatletter
\def\calign@preamble{%
   &\hfil\strut@
    \setboxz@h{\@lign$\m@th\displaystyle{##}$}%
    \ifmeasuring@\savefieldlength@\fi
    \set@field
    \hfil
    \tabskip\alignsep@
}
\let\cmeasure@\measure@
\patchcmd\cmeasure@{\divide\@tempcntb\tw@}{}{}{}
\patchcmd\cmeasure@{\divide\@tempcntb\tw@}{}{}{}
\patchcmd\cmeasure@{\ifodd\maxfields@
  \global\advance\maxfields@\@ne
  \fi}{}{}{}    
\newenvironment{calign}
{%
  \let\align@preamble\calign@preamble
  \let\measure@\cmeasure@
  \align
}
{%
  \endalign
}  
\makeatother
\tikzset{
    master/.style={
        execute at end picture={
            \coordinate (lower right) at (current bounding box.south east);
            \coordinate (upper left) at (current bounding box.north west);
        }
    },
    slave/.style={
        execute at end picture={
            \pgfresetboundingbox
            \path (upper left) rectangle (lower right);
        }
    }
}

\theoremstyle{plain} 

\newtheorem{theorem}{Theorem}[section]
\newtheorem{lemma}[theorem]{Lemma}
\newtheorem{corollary}[theorem]{Corollary}          
\newtheorem{proposition}[theorem]{Proposition}

\newtheorem*{theorem*}{Theorem}
\newtheorem*{proposition*}{Proposition}

\theoremstyle{definition} 
\newtheorem{definition}[theorem]{Definition}
\newtheorem{remark}[theorem]{Remark}

\newtheorem*{definition*}{Definition}

\theoremstyle{remark}  
\newtheorem{example}[theorem]{Example}

\newtheoremstyle{special_statement} 
        {\topskip}
        {\topskip}
        {\addtolength{\leftskip}{2.5em} \itshape }
        {}
        {\bfseries}
        {:}
        {.5em}
        {}
\theoremstyle{special_statement}

\DeclareMathOperator{\Hom}{Hom}
\DeclareMathOperator{\End}{End}

\newcommand{\id}{\mathrm{id}}

\newcommand{\Image}{\mathrm{Im}}
\newcommand{\Tr}{\mathrm{Tr}}

\newcommand{\Aut}{\ensuremath{\mathrm{Aut}}}

\newcommand{\Rep}{\mathrm{Rep}}

\newcommand{\Chan}{\mathrm{Chan}}
\newcommand{\Mod}{\mathrm{Mod}}

\newcommand{\Ker}{\mathrm{Ker}}

\newcommand{\CP}{\ensuremath{\mathrm{CP}}}

\newcommand{\Hilb}{\ensuremath{\mathrm{Hilb}}}
\newcommand{\TwoHilb}{\ensuremath{\mathrm{2Hilb}}}

\newcommand\supp{\mathrm{supp}}

\newcommand\ignore[1]{}

\newcommand{\QBij}{\ensuremath{\mathrm{QBij}}}

\newcommand{\F}{\ensuremath{\mathrm{SSFA}}}

\makeatletter
\DeclareFontFamily{OMX}{MnSymbolE}{}
\DeclareSymbolFont{MnLargeSymbols}{OMX}{MnSymbolE}{m}{n}
\SetSymbolFont{MnLargeSymbols}{bold}{OMX}{MnSymbolE}{b}{n}
\DeclareFontShape{OMX}{MnSymbolE}{m}{n}{
    <-6>  MnSymbolE5
   <6-7>  MnSymbolE6
   <7-8>  MnSymbolE7
   <8-9>  MnSymbolE8
   <9-10> MnSymbolE9
  <10-12> MnSymbolE10
  <12->   MnSymbolE12
}{}
\DeclareFontShape{OMX}{MnSymbolE}{b}{n}{
    <-6>  MnSymbolE-Bold5
   <6-7>  MnSymbolE-Bold6
   <7-8>  MnSymbolE-Bold7
   <8-9>  MnSymbolE-Bold8
   <9-10> MnSymbolE-Bold9
  <10-12> MnSymbolE-Bold10
  <12->   MnSymbolE-Bold12
}{}

\let\llangle\@undefined
\let\rrangle\@undefined
\DeclareMathDelimiter{\llangle}{\mathopen}%
                     {MnLargeSymbols}{'164}{MnLargeSymbols}{'164}
\DeclareMathDelimiter{\rrangle}{\mathclose}%
                     {MnLargeSymbols}{'171}{MnLargeSymbols}{'171}
\makeatother

\usepackage[utf8]{inputenc}
\usepackage[T1]{fontenc}
\usepackage{mathptmx}
\usepackage{etoolbox}

\makeatletter
\def\@email#1#2{%
 \endgroup
 \patchcmd{\titleblock@produce}
  {\frontmatter@RRAPformat}
  {\frontmatter@RRAPformat{\produce@RRAP{*#1\href{mailto:#2}{#2}}}\frontmatter@RRAPformat}
  {}{}
}%
\makeatother

\begin{document}

\title{Entanglement-invertible channels}
\author{Dominic Verdon\thanks{dominic.verdon@bristol.ac.uk}}
\affil{School of Mathematics, University of Bristol}

\date{\today}

\maketitle

\begin{abstract}
In a well-known result~(R. Werner, J. Phys. A, 34(35):7081, 2001), Werner classified all \emph{tight} quantum teleportation and dense coding schemes, showing that they correspond to \emph{unitary error bases}. Here tightness is a certain dimensional restriction: the quantum system to be teleported and the entangled resource must be of dimension $d$, and the measurement must have $d^2$ outcomes.  Here we generalise this classification so as to remove the dimensional restriction altogether, thereby resolving an open problem raised in that work. In fact, we classify not just teleportation and dense coding schemes, but \emph{entanglement-reversible channels}. These are channels between finite-dimensional $C^*$-algebras which are reversible with the aid of an entangled resource state, generalising ordinary reversibility of a channel. We show that such channels correspond to families of linear maps which are bi-isometric with respect to a duality defined by the resource state. In particular, in Werner's classification, a bijective correspondence between tight teleportation and dense coding schemes was shown: swapping Alice and Bob's operations turns a teleportation scheme into a dense coding scheme and vice versa. We observe that this property generalises ordinary invertibility of a channel; we call it \emph{entanglement-invertibility}. We show that entanglement-invertible channels are precisely the \emph{quantum bijections} previously studied in noncommutative topology~(B. Musto et al., J. Math. Phys, 59(8):081706, 2018), and therefore admit a classification in terms of Wang's quantum permutation group~(S. Wang, Comm. Math. Phys., 195:195-211, 1998).
\end{abstract}

\section{Introduction}

\subsection{A full classification of quantum teleportation and dense coding protocols}

Quantum teleportation and dense coding (sometimes called superdense coding) protocols are of vital importance in quantum computation and information theory. These protocols are in fact so foundational that we will assume the reader is familiar with them; if not, a nice exposition of the standard qubit teleportation and dense coding schemes can be found in the very first chapter of~\cite{Nielsen2010}. 

The papers in which quantum teleportation and dense coding were first defined~\cite{Bennett1992,Bennett1993} gave the standard qubit schemes; it was then natural to ask what other teleportation and dense coding schemes exist, and whether there is some classification of these schemes. In~\cite{Werner2001}, Werner gave a partial answer to this question; he restricted his attention to the \emph{tight} case. For teleportation this means that the state to be teleported is in a Hilbert space of dimension $d$, the shared entangled state is of two $d$-dimensional Hilbert spaces, and the measurement Alice performs has $d^2$ possible outcomes. Under these conditions, Werner showed that teleportation and dense coding schemes are in bijection with \emph{unitary error bases}, bases of unitary matrices orthogonal under the Hilbert-Schmidt inner product. That is, a unitary error basis yields both a tight teleportation scheme and a tight dense coding scheme, and all schemes are obtained this way. 

Despite the name of the paper~\cite{Werner2001}, this is only a classification of \emph{tight} teleportation and dense coding schemes, and it is natural to ask whether a classification exists for general schemes. Although there was some progress in this direction (in particular, we note~\cite{Albeverio2000,Albeverio2002,Mozes2005,Wu2006,Feng2006}) the question has 
not to our knowledge been solved before now.

The primary goal of this work is to classify teleportation and dense coding schemes in full generality. In fact, we will find that our classification extends more generally to \emph{entanglement-reversible} and \emph{entanglement-invertible channels}, which we will define shortly. The key tool we will use to state and prove the classification is an extended version of the graphical calculus of tensor network diagrams. It has long been observed (see e.g.~\cite{Abramsky2004,Heunen2019}) that the diagrammatic calculus is a convenient tool for studying quantum teleportation and dense coding, and here we find it absolutely essential; while it is in theory possible to state and prove our results without it, all intuition would be lost.

\subsection{Entanglement-reversibility and entanglement-invertibility}

We will now define what we mean by \emph{entanglement-reversible channels}. As is standard in quantum information theory, when we talk about \emph{channels} we mean completely positive trace-preserving maps between finite-dimensional (f.d.) $C^*$-algebras. 

Let $A,B$ be f.d. $C^*$-algebras. Recall that a channel $M: A \to B$ is \emph{reversible} if there exists a channel $N: B \to A$ such that $N \circ M = \id_A$; the channel $N$ is called a \emph{left inverse} for $M$. The channel is furthermore \emph{invertible} if $M \circ N = \id_B$; in this case $\dim(A) = \dim(B)$ and the left inverse $N$ is uniquely defined.

We generalise these definitions to account for an entangled resource state.  
\begin{definition}
Let $H_1,H_2$ be two Hilbert spaces, let $B(H_1)$ and $B(H_2)$ be the $C^*$-algebras of operators on these spaces and let $\sigma: B(H_1) \otimes B(H_2) \to B(H_2) \otimes B(H_1)$ be the swap channel. Let $W: \mathbb{C} \to B(H_1) \otimes B(H_2)$ be any channel (i.e. any state of $B(H_1) \otimes B(H_2)$). 

Let $M: A \otimes B(H_1) \to B$ be a channel. We say that $M$ is \emph{entanglement-reversible} with respect to $W$ if there exists a channel $N: B \otimes B(H_2) \to A$ satisfying the left equation of~\eqref{eq:defintro}. (The diagrams are read from bottom to top.) In this case we say that $N$ is an \emph{entanglement-left inverse of $M$ w.r.t. $W$}. If the right equation of~\eqref{eq:defintro} is additionally satisfied we say that $M$ is \emph{entanglement-invertible} with respect to $W$, and that $N$ is an \emph{entanglement-inverse} for $M$ w.r.t. $W$.
\begin{align}\nonumber
\includegraphics[valign=c]{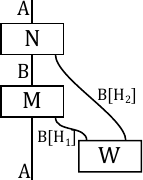}
~~=~~
\includegraphics[valign=c]{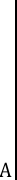}
&&
\includegraphics[valign=c]{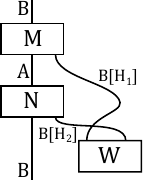}
~~=~~
\includegraphics[valign=c]{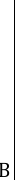}
\\\label{eq:defintro}
N \circ (M \otimes \id_{B(H_2)}) \circ  (\id_A \otimes W)
= 
\id_A
&&
M \circ (N \otimes \id_{B(H_1)}) \circ (\id_B \otimes \sigma) \circ (\id_B \otimes W)
=
\id_B
\end{align}
It is clear that these definitions reduce to ordinary reversibility and invertibility when $\dim(H_1) = \dim(H_2) = 1$.
\end{definition}
\begin{example}
The standard examples of entanglement-reversible channels are teleportation and dense coding schemes~\cite{Bennett1992,Bennett1993}. Let $K$ be some Hilbert space, and let $[n]$ be the $n$-dimensional commutative $C^*$-algebra. (Throughout we use the same notation for commutative $C^*$-algebras as for finite sets, since the two are equivalent by Gelfand duality.) Then:
\begin{itemize}
\item Let $A := B(K)$, and $B := [n]$. Then an entanglement-reversible channel $M: A \otimes B(H_1) \to B$ is precisely a quantum teleportation scheme. Using one half of the resource state $W$, a state $\sigma$ of the system $B(K)$ is transformed into classical information, from which $\sigma$ can be recovered using the other half of the resource state $W$.
\item Let $A:= [n]$, and let $B:= B(K)$. Then an entanglement-reversible channel $M: A \otimes B(H_1) \to B$ is precisely a quantum dense coding scheme. Using one half of the resource state $W$, some state in $i \in \{1,\dots,n\}$ is transformed into a quantum state $\omega_i \in B(K)$, from which $i$ can be recovered using the other half of the resource state $W$.
\end{itemize}
Of course, entanglement-reversibility is more general than teleportation and dense coding; we could consider entanglement-reversible classical-to-classical or quantum-to-quantum channels, for instance. 
\end{example}
\noindent
In~\cite[Thm. 1]{Werner2001}, Werner classified \emph{tight} teleportation and dense coding schemes. Tightness is a dimensional restriction: the Hilbert spaces $K,H_1,H_2$ all have the same dimension $d$, and one fixes $n:= d^2$. In this case it was shown that:
\begin{itemize}
\item For entanglement-reversibility, $W$ must be a maximally entangled pure state. 
\item Any entanglement-reversible channel is furthermore entanglement-invertible, yielding a bijective correspondence between tight teleportation and tight dense coding schemes. 
\item A tight teleportation or dense coding scheme is precisely specified by the data of a \emph{unitary error basis} (a basis of unitary operators in $B(K)$ orthogonal under the trace inner product).
\end{itemize}
In this work we extend Werner's classification to general entanglement-reversible channels, without any dimensional restriction.

\subsection{Results}\label{sec:summary}

\subsubsection{Classification of entanglement-reversible and entanglement-invertible channels}

To obtain this classification we use the notions of \emph{bi-isometry} and \emph{minimal dilation}, which we will shortly define. We also use the graphical calculus of shaded tensor network diagrams~\cite{Reutter2019}, where diagrams represent indexed families of linear maps. In this short summary of the results we will not use any tensor network diagrams, since we have not introduced the graphical calculus yet; however, the reader can find the diagrammatic statements of these results by consulting the statements in the body of the paper.

We call an indexed family of Hilbert spaces a \emph{1-morphism} and an indexed family of linear maps a \emph{2-morphism}. (The language of 1- and 2-morphisms reflects the fact that the shaded calculus is the graphical calculus of $\TwoHilb$, the semisimple $C^*$-2-category of finite-dimensional 2-Hilbert spaces and linear maps~\cite{Vicary2012,Verdon2021}. However, no category theory is required in order to understand our results; we give a full introduction to the shaded calculus which does not mention categories.) There are notions of tensor product, duality and dimension for 1-morphisms, and Hermitian adjunction for 2-morphisms, extending the corresponding notions for single Hilbert spaces and linear maps. A dual for a 1-morphism $X$ is a triple $(X^*,\epsilon,\eta)$, where $X^*$ is a dual 1-morphism and $\epsilon$ and $\eta$ are `cap' and `cup' 2-morphisms defining the duality; there is always a \emph{standard} choice of dual, which is unique up to unitary isomorphism. 

It is well-known that every f.d. $C^*$-algebra $A$ decomposes as a multimatrix algebra $A \cong \bigoplus_{i\in I} B(H_i)$ for some Hilbert spaces $\{H_i\}_{i \in I}$, where $I$ is a finite index set; this indexed family of Hilbert spaces defines a 1-morphism, which we will call $X_A$.  
\begin{definition*}
Let $A,B$ be f.d. $C^*$-algebras. By a generalised Stinespring's theorem (Theorem~\ref{thm:stinespring}), a channel $F: A \otimes B(H) \to B$ corresponds to a family of \emph{dilations} $\{(E,\tau)\}$, where $E$ is a 1-morphism and $\tau: H \otimes X_A \to X_B \otimes E$ is a 2-morphism. The dilation minimising $\dim(E)$ is unique up to unitary isomorphism, and we call it the \emph{minimal dilation} of the channel. 
\end{definition*}
\begin{definition*}
Let $\tau: H \otimes X_A \to X_B \otimes E$ be a 2-morphism and let $(H^*,\eta,\epsilon)$, $(E^*,\eta,\epsilon)$ be duals for $H$ and $E$. Let $\tau^T: X_A \otimes E^* \to H^* \otimes X_B$ be the partial transpose with respect to these duals. We say that $\tau$ is a \emph{bi-isometry} with respect to these duals if $\tau^{\dagger} \circ \tau = \id_{H \otimes X_A}$ and $(\tau^{T})^{\dagger} \circ \tau^T = \id_{X_A \otimes E^*}$. We say that $\tau$ is a \emph{biunitary} with respect to these duals if it is a bi-isometry and additionally $\tau \circ \tau^{\dagger} = \id_{X_B \otimes E}$ and $\tau^T \circ (\tau^T)^{\dagger} = \id_{H^* \otimes X_B}$.
\end{definition*}
\noindent
We first obtain a classification of channels entanglement-invertible w.r.t. the maximally entangled pure state, in terms of biunitary 2-morphisms. 
\begin{proposition*}[Proposition~\ref{prop:qbijmindil}]
A channel $F: A \otimes B(H) \to B$ is entanglement-invertible w.r.t. the maximally entangled pure state of $B(H) \otimes B(H)$ if and only if its minimal dilation $\tau: H \otimes X_A \to X_B \otimes E$ is (up to normalisation) biunitary w.r.t. the standard duality on $H$ and $E$. In particular, this implies $\dim(A) = \dim(B)$.
\end{proposition*}
\noindent
Channels whose minimal dilation is a biunitary were called \emph{quantum bijections} in~\cite[Def. 4.3]{Musto2018}; our results therefore give an operational interpretation of this mathematical definition. As was shown in~\cite{Musto2019}, quantum bijections possess a nice compositional structure; a map between quantum bijections is called an \emph{intertwiner}.

We now proceed to classify entanglement-reversible and entanglement channels w.r.t. a fixed resource state $W$ in terms of their minimal dilation. Note that for any linear map $\omega: H \to H$ satisfying $\Tr(\omega^{\dagger} \omega) \neq 0$ one can define a corresponding pure state of $H \otimes H$ by normalising $(\omega \otimes \mathbbm{1}) \ket{\Phi}$, where $\ket{\Phi} \in H \otimes H$ is the canonical maximally entangled state. We say that this is the `state defined by the map $\omega$'.
\begin{theorem*}[Theorem~\ref{thm:genpure}]
Let $W: \mathbb{C} \to B(H) \otimes B(H)$ be a pure state defined by an invertible map $\omega: H \to H$, and let $M: A \otimes B(H) \to B$ be a channel. Then:
\begin{itemize}
\item The channel $M$ is entanglement-reversible w.r.t $W$ precisely when its minimal dilation is a bi-isometry w.r.t. a certain duality defined by the state $W$. In particular, it is a necessary condition that $\dim(A) \leq \dim(B)$.
\item If additionally $\dim(A) = \dim(B)$, then the minimal dilation of $M$ is furthermore a biunitary w.r.t. the duality defined by $W$. Moreover, the entanglement-left inverse is uniquely defined. 
\item The channel $M$ is furthermore entanglement-invertible w.r.t. $W$ precisely when the following conditions are satisfied:
\begin{itemize}
\item $M$ is a quantum bijection.
\item The map $\omega^{\dagger} \circ \omega$ is an intertwiner of quantum bijections $M \to M$.
\end{itemize}
\end{itemize}
\end{theorem*}
\noindent
Once this theorem is proved it is straightforward to extend the result to general pure and mixed states $W \in B(H_1) \otimes B(H_2)$, since any mixed state can be decomposed as a convex combination of pure states and, up to a quotient and an injection, all of these pure states are defined by an invertible map; the general result is given as Corollary~\ref{cor:genmixed}. 

Finally, we show in Section~\ref{sec:wernerex} how Werner's classification of tight teleportation and dense coding schemes in terms of unitary error bases emerges straightforwardly from our more general result. We expect similar methods can be used to extract concrete descriptions of entanglement-reversible channels in other special cases.
\subsection{Related work}

\paragraph{Teleportation and dense coding outside of the tight scenario.} We highlight some relevant previous work on this problem; this list is not exhaustive. With regard to dense coding: the papers~\cite{Shadman2010,Situ2013} dealt with superdense coding over noisy quantum channels or with noisy encoding operations; this can be brought within our framework by considering entanglement-reversibility of $N \circ M$ w.r.t. a state $W$, where $M$ is the encoding channel and $N$ is a channel representing the noise.  The papers~\cite{Mozes2005,Wu2006} provide dimensional bounds for dense coding with arbitrary entangled pure state $W$. The paper~\cite{Feng2006} studies tight dense coding with an arbitrary entangled pure state $W$, in the case where some nonzero probability of failure is allowed.  With regard to teleportation: the papers~\cite{Albeverio2000,Albeverio2002} give conditions for entanglement-reversibility of a channel $(M,H): B(K) \to{} [d]$ w.r.t. a general pure state $W$ with no dimensional restriction when the channel $M$ is a complete projective measurement. 

\paragraph{Categorical quantum mechanics.} This work makes use of the technology of categorical quantum mechanics, in particular the 2-categorical diagrammatic calculus~\cite{Selinger2010} which was applied to quantum mechanics in~\cite{Vicary2012,Vicary2012a} and further developed in~\cite{Reutter2019,Heunen2019}; we also use the covariant Stinespring theorem~\cite[Thm. 4.9]{Verdon2021} (although since there is no symmetry group here the special case in this paper also follows straightforwardly from~\cite[Cor. 4.13]{Selinger2007}). Our treatment of `splitting' finite-dimensional $C^*$-algebras is based on Q-system completion for rigid $C^*$-tensor categories~\cite{Chen2022}. 

However, no category theory is required in order to understand this paper; in particular, we present an introduction to the diagrammatic calculus without ever referring to categories. 

\subsection{Acknowledgments}
We are grateful to David Reutter and Jamie Vicary for useful discussions. We thank Ashley Montanaro for his support of this work. This work has been funded by the European Research Council (ERC) under the European Union’s Horizon 2020 research and innovation programme (grant agreement No. 817581). This work has also been funded by EPSRC.

\section{Background}

\subsection{Diagrammatic calculus}

In this work we will make use of a diagrammatic calculus of shaded tensor network diagrams, which has appeared before in~\cite{Heunen2019,Reutter2019}. We now provide an elementary introduction which requires no background in category theory.

\subsubsection{The unshaded calculus}\label{sec:unshaded}
We will first review the well-known tensor network diagram calculus (see e.g.~\cite{Selinger2010,Heunen2019}), in which wires correspond to finite-dimensional Hilbert spaces and boxes correspond to linear maps.  \ignore{The horizontal direction in the diagram corresponds to monoidal product and the vertical direction to composition.} We read diagrams from bottom to top, so input wires come in from the bottom and output wires exit at the top of the diagram.\ignore{ For example, for $f: V_1 \to V_2$ and  $g:V_2 \to V_3$, the  composition $g \circ f: V_1 \to V_3$ and monoidal product $f \otimes g: V_1 \otimes  V_2 \to V_2 \otimes V_3$ are depicted as follows:}
Composition and tensor product are depicted by vertical and horizontal juxtaposition respectively. For instance, let $f: V_1 \to V_2$ and $g: V_2 \to V_3$ be linear maps; then they can be composed as follows:
\begin{align*}
\includegraphics{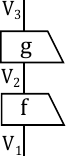}
&&
\includegraphics{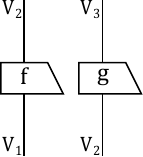}
\\
g \circ f: V_1 \to V_3 && f \otimes g: V_1 \otimes V_2 \to V_2 \otimes V_3 
\end{align*}
The reader will notice the boxes have an offset edge; this is so we can represent the transpose, dagger and complex conjugate of a linear map, as we will discuss shortly.

Wires corresponding to the one-dimensional Hilbert space $\mathbb{C}$ are not drawn. A diagram with no input and no output wires therefore represents a linear map $\mathbb{C} \to \mathbb{C}$, i.e. a scalar. Likewise, a diagram with no input wires represents a linear map $\psi: \mathbb{C} \to V$, where $V$ is the Hilbert space specified by its output wires; such linear maps obviously correspond to vectors $\ket{\psi} \in V$, where $\ket{\psi}:= \psi(1)$. Likewise, a diagram with no output wires represents a vector $\bra{\psi} \in V^*$, where $V$ is the Hilbert space specified by the input wires of the diagram. From now on we will use the braket notation for both the vector and the associated linear map, so we will write (for instance) $\ket{\psi}: \mathbb{C} \to V$.

Every f.d. Hilbert space $V$ is self-dual. (We remark that using the self-duality of f.d. Hilbert spaces in the way we do here is \emph{evil} in the sense of category theory, since it relies on a unnatural choice of orthonormal basis for each Hilbert space. We find that it simplifies the graphical calculus, but those who shun evil will prefer to distinguish between a Hilbert space and its dual.) Let $\{\ket{i}\}$ be some orthonormal basis of $V$, and let 
\begin{align}\label{eq:etadef}
\ket{\eta_V}:= \sum_{i =1}^{d}\ket{i} \otimes \ket{i} \in V \otimes V.\end{align} 
(We call the normalisation $\frac{1}{\sqrt{\dim(V)}} \ket{\eta_V}$ the \emph{canonical maximally entangled state} of $V \otimes V$; it is sometimes known as the Bell state.) Then the self-duality of $V$ is characterized by the vectors $\ket{\eta_V} \in V \otimes V$ and $\bra{\eta_V} \in (V \otimes V)^*$; in the graphical calculus we represent these linear maps topologically as \textit{cups and caps}:
\begin{calign}\label{eq:cupscapsHilb}
\includegraphics[scale=0.75,valign=c]{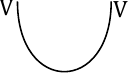}
&
\includegraphics[scale=0.75,valign=c]{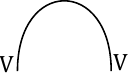}
\\
\ket{\eta_V}: \mathbb{C} \to V \otimes V
& 
\bra{\eta_V}: V \otimes V \to \mathbb{C} 
\end{calign}
These maps fulfill the following \textit{snake equations}:
\begin{calign}\label{eq:snake}
\includegraphics[scale=0.75,valign=c]{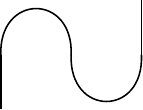}
~~~=~~~
\includegraphics[scale=0.75,valign=c]{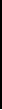}
~~~= ~~~
\includegraphics[scale=0.75,valign=c]{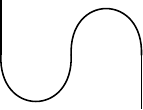}
\end{calign}
Together with the swap map $\sigma_{V,W}:\ignore{V\otimes W\to W\otimes V$,~$}v\otimes w\mapsto w\otimes v$, depicted as a crossing of wires, this leads to an extremely flexible topological calculus, in which we can untangle arbitrary diagrams and straighten out any twists:
\begin{calign}
\includegraphics[scale=1,valign=c]{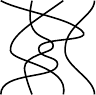}
~~=~~
\includegraphics[scale=1,valign=c]{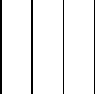}
&
\includegraphics[scale=1,valign=c]{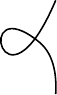}
~~~=~~~\includegraphics[scale=1,valign=c]{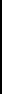}
~~~=~~~
\includegraphics[scale=1,valign=c]{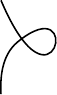}
\ignore{
\input{unwindtwists.tikz}}
\end{calign}
\ignore{We will refer to linear maps of the form $\C \to V$ and $V \to\C$ for a Hilbert space $V$ as \emph{states} and \emph{effects}, where the more traditional definition of  a state as a vector $\ket{\psi} \in V$ is retrieved as $\ket{\psi} = \psi(1)$.}
Let $D$ be a diagram representing a linear map $f:V_1 \otimes \dots \otimes V_m  \to W_1 \otimes \dots \otimes W_n$ between Hilbert spaces. Then the Hermitian adjoint (colloquially, the \emph{dagger}) \mbox{$f^\dagger:W_1 \otimes \dots \otimes W_n \to V_1 \otimes \dots \otimes V_{m}$} is represented by the reflection of the diagram $D$ across a horizontal axis:
\begin{align*}
\includegraphics[scale=0.75,valign=c]{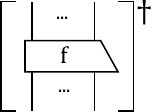}
~~=:~~
\includegraphics[scale=0.75,valign=c]{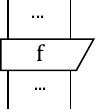}
\end{align*}
The transpose $f^T: W_n \otimes \dots \otimes W_1 \to V_m \otimes \dots \otimes V_1$ with respect to the orthonormal basis defining the self-duality is represented by means of a $\pi$-rotation of the corresponding diagram:
\begin{align*}
\includegraphics[scale=0.75,valign=c]{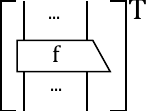}
~~=~~
\includegraphics[scale=0.75,valign=c]{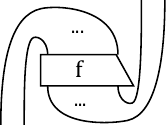}
~~=~~
\includegraphics[scale=0.75,valign=c]{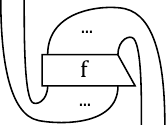}
~~=:~~
\includegraphics[scale=0.75,valign=c]{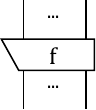}
\end{align*}
Finally, the complex conjugate $f^*: V_m \otimes \dots \otimes V_1 \to W_n \otimes \dots \otimes W_1$ with respect to the orthonormal basis defining the self-duality is represented by means of a reflection in a vertical axis:
\begin{align*}
\includegraphics[scale=0.75,valign=c]{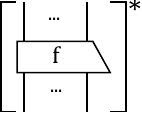}
~~=~~
\includegraphics[scale=0.75,valign=c]{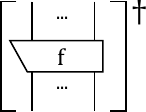}
~~=~~
\includegraphics[scale=0.75,valign=c]{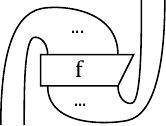}
~~=:~~
\includegraphics[scale=0.75,valign=c]{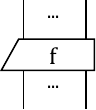}
\end{align*}
With this notation, the boxes slide along the wires as one would expect (where the cups and caps in the diagrams are those of~\eqref{eq:cupscapsHilb}:
\begin{align*}
\includegraphics[valign=c]{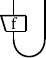}
~~=~~
\includegraphics[valign=c]{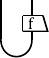}
&&
\includegraphics[valign=c]{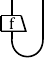}
~~=~~
\includegraphics[valign=c]{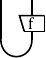}
&&
\includegraphics[valign=c]{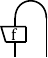}
~~=~~
\includegraphics[valign=c]{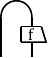}
&&
\includegraphics[valign=c]{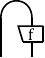}
~~=~~
\includegraphics[valign=c]{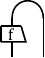}
\\[5pt]
\includegraphics[valign=c]{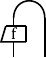}
~~=~~
\includegraphics[valign=c]{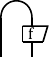}
&&
\includegraphics[valign=c]{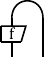}
~~=~~
\includegraphics[valign=c]{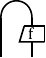}
&&
\includegraphics[valign=c]{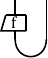}
~~=~~
\includegraphics[valign=c]{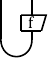}
&&
\includegraphics[valign=c]{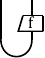}
~~=~~
\includegraphics[valign=c]{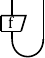}
\end{align*}
For any Hilbert space $V$, we note the following expression for the trace of a linear map $f \in \End(V)$, and in particular for the dimension $\dim(V) = \Tr(\id_V)$:
\begin{align}\label{eq:tracedef}
\Tr(f)~~=~~\includegraphics[valign=c]{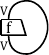}
&&
\dim(V)~~=~~ \includegraphics[valign=c]{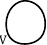}
\end{align}
\subsubsection{The shaded calculus}

We now extend to the graphical calculus of shaded tensor network diagrams.  Formally, this is the graphical calculus of the semisimple rigid $C^*$-2-category $\TwoHilb$ of finite-dimensional 2-Hilbert spaces. In this work, however, we will avoid category theory altogether and introduce the shaded calculus simply as an indexed version of the unshaded calculus. 

\paragraph{Wires and boxes.}
In the shaded calculus, the regions in a tensor network diagram can be shaded. These shaded regions correspond to finite index sets. Wires now correspond to \emph{families} of Hilbert spaces, indexed by the parameters of the regions to the left and right of the wire. We call these indexed families of Hilbert spaces \emph{1-morphisms}. For example, let $[m]$ and $[n]$ be two index sets. From now on we shade regions corresponding to the set $[m]$ with wavy lines and regions corresponding to the set $[n]$ with polka dots. Consider the following wire:
\begin{align*}
\includegraphics{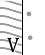}
\end{align*}
We see that an $[m]$-region is on the left of the $V$-wire and an $[n]$-region on the right. Hence the wire $V$ is an $[m] \times [n]$-indexed family of Hilbert spaces; these can be arranged in an $[m] \times [n]$ matrix $(V_{ij})_{(i,j) \in [m] \times [n]}$.  We write $V: [m] \to{} [n]$ to indicate that the $[m]$-region is on the left and the $[n]$-region on the right of the wire $V$.

Boxes now correspond to \emph{families} of linear maps, which are indexed by the parameters of the adjoining regions. This is best explained by example. Here are two boxes:
\begin{align}
\label{eq:twoboxes}
\includegraphics{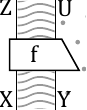}
&&
\includegraphics{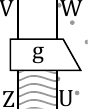}
\end{align}
On the left, let us look at the wires first; we see that $X = (X_m)_{m \in [m]}$ and $Z = (Z_m)_{z \in [m]}$ are $[m]$-indexed families of Hilbert spaces, and $U = (U_{mn})_{(m,n) \in [m] \times [n]}$ and $Y=(Y_{mn})_{(m,n) \in [m] \times [n]}$ are $[m]\times[n]$-indexed families of Hilbert spaces. The morphism $f: X \otimes Y \to Z \otimes U$ is a family of linear maps $\{f_{ijk}\}_{(i,j,k) \in [m] \times  [n] \times [m]}$, where the indices correspond to the regions to the bottom, the right and the top of the box, in that order. With this indexing, we see that the map $f_{ijk}$ has type $X_{i} \otimes Y_{ij} \to Z_{k} \otimes U_{kj}$. 

On the right, we see similarly that $V$ is a single Hilbert space (there being no adjacent shaded regions) and $W = (W_n)_{n \in [n]}$ is an $[n]$-indexed family of Hilbert spaces. Then $g: Z \otimes U \to V \otimes W$ is an $[m] \times [n]$-indexed family of linear maps $\{g_{ij}\}_{(i,j) \in [m] \times [n]}$, where $g_{ij}: Z_i \otimes U_{ij} \to V \otimes W_j$.

We call an indexed family of linear maps a \emph{2-morphism}.

\paragraph{Composition.} 
We can compose boxes to create new 2-morphisms. We refer to a general planar diagram of wires and boxes as a \emph{2-morphism diagram}. The family of linear maps represented by a 2-morphism diagram is indexed by the parameters of the open regions, while the closed regions are summed over. Composition is given by vertical juxtaposition, as in the unshaded case. Again, an example is probably the best explanation.  Looking at the 2-morphisms~\eqref{eq:twoboxes}, we see that the output 1-morphism of $f$ is the same as the input 1-morphism of $g$, so we can form the composite:
\begin{align*}
\includegraphics{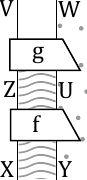}
\end{align*}
Let us compute this composite. Looking at the open regions in the diagram, we see that $g \circ f$ is an $[m] \times [n]$-indexed family of linear maps, where the first index corresponds to the bottom open shaded region and the second to the open shaded region on the right. There is one closed region, whose index is to be summed over. We therefore see that these linear maps are defined as follows:
$$(g \circ f)_{ij} = \sum_{k =1}^{m} g_{kj} \circ f_{ijk}$$

\paragraph{Identity wires.} Recall that, in the unshaded calculus, the wire corresponding to the one-dimensional Hilbert space is invisible. In the shaded calculus, for every index set $[m]$, there is a canonical \emph{identity} 1-morphism $\id_{[m]}: [m] \to{} [m]$, specified by the following $[m] \times [m]$ matrix of Hilbert spaces:
$$(\id_{[m]})_{ij} = 
\begin{cases}
{\bf 0} \quad & i \neq j
\\
\mathbb{C} \quad & i=j
\end{cases}$$ 
This wire is invisible in the shaded calculus. We draw boxes $\alpha: \id_{[m]} \to \id_{[m]}$ as discs surrounded by a dotted line:
\begin{align}\label{eq:floatingdiscs}
\includegraphics{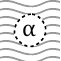}
\end{align}
These discs may be moved around freely inside their containing region. For each value of the index set associated to their region, they specify a scalar.

\paragraph{Duality.}
We now extend duality to the shaded setting. Let $V= (V_{ij})_{(i,j) \in [m] \times [n]}$ be a family of Hilbert spaces indexed by $[m]$ on the left and $[n]$ on the right. Then we define the \emph{dual} $V^*$ to be the family $V^* = (V_{ji})_{(i,j) \in [n] \times [m]}$ indexed by $[n]$ on the left and $[m]$ on the right. In the diagrammatic calculus we draw a wire $V$ with an upwards-facing arrow and its dual $V^*$ with a downwards facing arrow. We now define cup and cap morphisms generalising~\eqref{eq:cupscapsHilb}, depicted as follows: 
\begin{align*}
\includegraphics{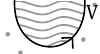}
&&
\includegraphics{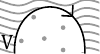}
&&
\includegraphics{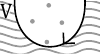}
&&
\includegraphics{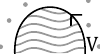}
\\
\eta_V: \id_{[n]} \to V^* \otimes V
&&
\epsilon_V:  V \otimes V^* \to \id_{[m]}
&&
\epsilon_V^{\dagger}: \id_{[m]} \to V \otimes V^* 
&&
\eta_V^{\dagger}: V^* \otimes V \to \id_{[n]}
\end{align*}
Let us first define $\eta$. Drawing in the invisible input wire $\id_{[n]}$, we see that $\eta= (\eta_{ijk})_{(i,j,k) \in [n] \times [m] \times [n]}$, where $\eta_{ijk}: (\id_{[n]})_{ik} \to V_{ji} \otimes V_{jk}$. Clearly if $i \neq k$ then $\eta_{ijk}$ must be the zero morphism, since $(\id_{[n]})_{ik}$ is the zero Hilbert space. If $i = k$ then we define $\eta_{iji} = \ket{\eta_{V_{ji}}}: \mathbb{C} \to V_{ji} \otimes V_{ji}$, recalling the definition of $\ket{\eta_{V_{ji}}}$ from~\eqref{eq:etadef}.

We define $\epsilon$ similarly. Drawing in the invisible output wire $\id_{[m]}$, we see that $\epsilon = (\epsilon_{ijk})_{(i,j,k) \in [m] \times [n] \times [m]}$, where $\epsilon_{ijk}: V_{ij} \otimes V_{kj} \to (\id_{[m]})_{ik}$. Again, if $i \neq k$ then $\epsilon_{ijk}$ must be the zero morphism, since $(\id_{[m]})_{ik}$ is the zero Hilbert space. If $i=k$ then we define $\epsilon_{iji} = \bra{\eta_{V_{ij}}}: V_{ij} \otimes V_{ij} \to \mathbb{C}$.

The 2-morphisms $\epsilon^{\dagger}$ and $\eta^{\dagger}$ are defined similarly. Alternatively, they can be defined as the daggers of the 2-morphisms $\eta$ and $\epsilon$; the dagger will be defined in the next paragraph. We call $\eta$ and $\epsilon$ the \emph{right} cup and cap (since the arrow goes from left to right) and $\epsilon^{\dagger}$ and $\eta^{\dagger}$ the \emph{left} cup and cap. It is straightforward to check that the 2-morphisms $\eta, \epsilon, \eta^{\dagger}, \epsilon^{\dagger}$ obey the following snake equations:
\begin{align}\label{eq:shadedsnake}
\includegraphics[valign=c]{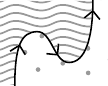}
~~=~~
\includegraphics[valign=c]{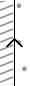}
~~=~~
\includegraphics[valign=c]{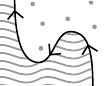}
&&
\includegraphics[valign=c]{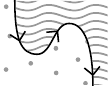}
~~=~~
\includegraphics[valign=c]{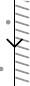}
~~=~~
\includegraphics[valign=c]{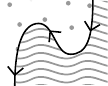}
\end{align}
We are therefore able to deform wires topologically as desired. There is a swap map in this calculus (see~\cite{Reutter2019,Heunen2019} for more details), but we will not use it in this paper.

\paragraph{Dagger, transpose and conjugate.} The notions of dagger, transposition and complex conjugation extend straightforwardly to the shaded calculus. 

\begin{itemize}
\item Let $f: X_1 \otimes \cdots \otimes X_m \to Y_1 \otimes \cdots \otimes Y_n$ be a box. This box represents an indexed family of linear maps. The \emph{dagger} of $f$ is the 2-morphism $f^{\dagger}: Y_1 \otimes \cdots \otimes Y_n \to X_1 \otimes \cdots \otimes X_m$ specified by taking the dagger of each linear map in the family for every choice of the indices. The dagger is represented by reflecting the diagram containing the box in a horizontal axis, so that the offset corner is at the top right, while preserving the orientation of the arrows on the wires.

Again, this is best illuminated by an example. Recall the box $g: Z \otimes U \to V \otimes W$ from~\eqref{eq:twoboxes}. The box $g^{\dagger}: V \otimes W \to Z \otimes U$ is depicted as follows:
\begin{align*}
\includegraphics{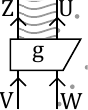}
\end{align*}
In our notation from before, $Z = (Z_m)_{m \in [m]}$, $W = (W_n)_{n \in [n]}$, and $U = (U_{mn})_{(m,n) \in [m] \times [n]}$, and $g = (g_{ij})_{(i,j) \in [m] \times [n]}$, where $g_{ij}: Z_i \otimes U_{ij} \to V \otimes W_j$. Now $g^{\dagger}$ is also an $[m] \times [n]$-indexed family, where now the region $[m]$ is above the box. So $g^{\dagger} = ((g^{\dagger})_{ij})_{(i,j) \in [m] \times [n]}$. Then $g^{\dagger}$ is defined by setting $(g^{\dagger})_{ij}:= (g_{ij})^{\dagger}$.

We extend the dagger to general 2-morphism diagrams by flipping the whole diagram in a horizontal axis, while preserving the orientation of any arrows. This is consistent, in the sense that the resulting family of linear maps can be obtained either by computing the composition associated to the flipped diagram, or equivalently  by taking the dagger of each of the linear maps associated to the original diagram.

\item Let $f: X_1 \otimes \cdots \otimes X_m \to Y_1 \otimes \cdots \otimes  Y_n$ be a box. The \emph{transpose} of the $f$ is a box $f^T: Y_n^{*} \otimes \cdots \otimes Y_1^{*} \to X_{m}^* \otimes \cdots \otimes X_1^{*}$, represented by a $\pi$-rotation of the box $f$, and defined using the duality as follows:
\begin{align}\label{eq:transposedef}
\includegraphics[valign=c]{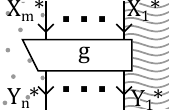}
:=
\includegraphics[scale=0.9,valign=c]{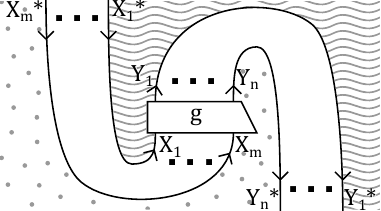}
=
\includegraphics[scale=0.9,valign=c]{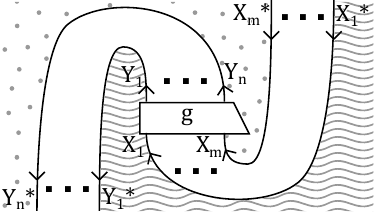}
\end{align}
This transpose may equivalently be defined as the componentwise transpose; that is, for each value of the indices, one takes the transpose of the corresponding linear map. The equality between the left and right transpose in~\eqref{eq:transposedef} therefore follows immediately from the equality of the left and right transpose in the unshaded calculus. 

\item Let $f: X_1 \otimes \cdots \otimes X_m \to Y_1 \otimes \cdots \otimes Y_n$ be a box. The \emph{complex conjugate} of $f$ is a box $f^*: X_m^{*} \otimes \cdots \otimes X_1^* \to Y_n^* \otimes \cdots \otimes Y_1^*$, represented by flipping the box $f$ in a vertical axis and reversing the orientations of the wires. It is defined as the dagger of the transpose, or equivalently the transpose of the dagger:
\begin{align*}
\includegraphics[valign=c]{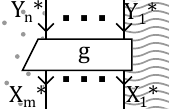}
~~:=~~
\includegraphics[valign=c]{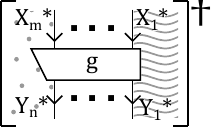}
~~=~~
\includegraphics[valign=c]{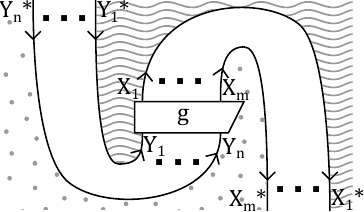}
\end{align*}
This may equivalently be defined as the componentwise complex conjugate, defined by taking the complex conjugate of the linear map corresponding to each value of the indices. 
\end{itemize}
With these definitions, the boxes slide around the wires as in the unshaded calculus:
\begin{align}\nonumber
\includegraphics[valign=c]{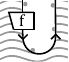}
~~=~~
\includegraphics[valign=c]{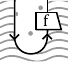}
&&
\includegraphics[valign=c]{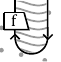}
~~=~~
\includegraphics[valign=c]{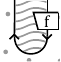}
&&
\includegraphics[valign=c]{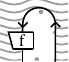}
~~=~~
\includegraphics[valign=c]{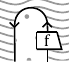}
&&
\includegraphics[valign=c]{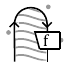}
~~=~~
\includegraphics[valign=c]{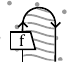}
\\[5pt]
\label{eq:shadedslide}
\includegraphics[valign=c]{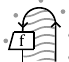}
~~=~~
\includegraphics[valign=c]{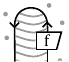}
&&
\includegraphics[valign=c]{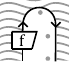}
~~=~~
\includegraphics[valign=c]{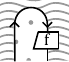}
&&
\includegraphics[valign=c]{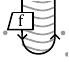}
~~=~~
\includegraphics[valign=c]{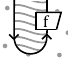}
&&
\includegraphics[valign=c]{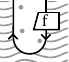}
~~=~~
\includegraphics[valign=c]{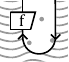}
\end{align}
In what follows we will use~\eqref{eq:shadedslide} together with the snake equations to deform and manipulate diagrams topologically. We will say that equalities arrived at in this way are `by isotopy of the diagram'.

\paragraph{Linear structure and endomorphism $C^*$-algebras.}

Consider the 2-morphism $f$ defined in~\eqref{eq:twoboxes}, with type $X \otimes Y \to Z \otimes U$. We observed above that it corresponds to a family of linear maps $\{f_{ijk}\}_{(i,j,k) \in [m] \times [n] \times [m]}$, where $f_{ijk}: X_i \otimes Y_{ij} \to Z_{k} \otimes U_{kj}$. 

Let $\Hom(X \otimes Y, Y \otimes Z)$ be the set of all 2-morphisms $X \otimes Y \to Y \otimes Z$. Such a 2-morphism is specified by a choice of linear map for each value of the indices $(i,j,k) \in [m] \times [n] \times [m]$. We therefore observe that 
\begin{align}\label{eq:homdecomp}
\Hom(X \otimes Y, Y \otimes Z) = \bigoplus_{(i,j,k) \in [m] \times [n] \times [m]} \Hom(X_{i} \otimes Y_{ij}, Z_{k} \otimes U_{kj}).
\end{align} 
The set $\Hom(X \otimes Y, Y \otimes Z)$ thereby acquires the structure of a Banach space; scalar multiplication and summation are defined componentwise, and the norm is the sum of the norms for each of the factors. This observation generalises in the obvious way to $\Hom(X_1 \otimes \cdots \otimes X_m, Y_1 \otimes \cdots \otimes Y_n)$, where $X_1 \otimes \cdots \otimes X_m$ and $Y_1 \otimes \cdots \otimes Y_n$ are any choice of input and output wires. In particular, we can consider sums and scalar multiples of  2-morphisms, which we will indicate by writing the diagrams as terms in algebraic expressions. 

The dagger $\dagger: \Hom(X_1 \otimes \cdots \otimes X_m, Y_1 \otimes \cdots \otimes Y_n) \to \Hom(Y_1 \otimes \cdots \otimes Y_n,X_1 \otimes \cdots \otimes X_m)$ which was defined above is just the componentwise dagger with respect to the decomposition~\eqref{eq:homdecomp}. In particular, it satisfies $||f^{\dagger} \circ f|| = ||f||^2$, and it follows that the endomorphism algebra $\End(X_1 \otimes \cdots \otimes X_m):= \Hom(X_1 \otimes \cdots \otimes X_m, X_1 \otimes \cdots \otimes X_m)$ is a finite-dimensional $C^*$-algebra, where the involution is given by the dagger. 

We note two facts about these endomorphism $C^*$-algebras:
\begin{itemize}
\item Let $f: X_1 \otimes \cdots \otimes X_m \to Y_1 \otimes \cdots \otimes Y_n$ be any 2-morphism. Then $f^{\dagger} \circ f$ is a positive element of the $C^*$-algebra $\End(X_1 \otimes \cdots \otimes X_m)$.
\item For any index set $[m]$, $\End(\id_{[m]})$ is a commutative $C^*$-algebra. (This fact is clear from the graphical calculus; since the endomorphisms are represented by floating discs~\eqref{eq:floatingdiscs} we can simply move one round the other.)
\end{itemize}

\paragraph{Isometries, unitaries, projection and partial isometries.} These notions generalise straightforwardly to 2-morphisms. Let $f:  X_1 \otimes \cdots \otimes X_m \to Y_1 \otimes \cdots \otimes Y_n$ be a 2-morphism. We say that $f$ is:
\begin{itemize}
\item An \emph{isometry} if $f^{\dagger} \circ f = \mathbbm{1}_{X_1 \otimes \cdots \otimes X_m}$.
\item A \emph{coisometry} if $f \circ f^{\dagger} = \mathbbm{1}_{Y_1 \otimes \cdots \otimes Y_n}$. 
\item A \emph{unitary} if it is both an isometry and a coisometry.
\item A \emph{partial isometry} if $f^{\dagger} \circ f \in \End(X_1 \otimes \cdots \otimes X_m)$ is a projection (equivalently, if $f \circ f^{\dagger} \in \End(Y_1 \otimes \cdots \otimes Y_n)$ is a projection).
\end{itemize}

\paragraph{Left dimension.}
For any wire $X: [m] \to{} [n]$, we define the \emph{left dimension} $d_{X} \in \End(\id_{[n]})$ as follows:
\begin{align*}
\includegraphics[valign=c]{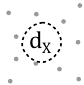}
~~:=~~
\includegraphics[valign=c]{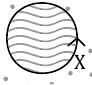}
\end{align*}
We observe that $d_X = \eta_{X}^{\dagger} \circ \eta_X$. In particular, by the first fact about endomorphism $C^*$-algebras noted above, it is a positive element of $\End(\id_{[n]})$. We write $n_X:= \sqrt{d_X} \in \End(\id_{[n]})$ for the positive square root of the left dimension. We assume throughout without loss of generality that $d_X$ and $n_X$ are invertible.

\paragraph{More general dualities.} Let $V: [m] \to{} [n]$ be a 1-morphism. Above we defined the canonical dual $V^*: [n] \to{} [m]$, together with cup and cap 2-morphisms $\eta_V: \id_{[n]} \to V^* \otimes V$ and $\epsilon_V: V \otimes V^* \to \id_{[m]}$ obeying the snake equations~\eqref{eq:shadedsnake}.

In fact, we can make a more general definition. We say that a 1-morphism $\overline{V^*}: [n] \to{} [m]$ is a \emph{dual} for $V$ if there exist cup and cap morphisms $\overline{\eta_V}: \id_{[n]} \to \overline{V^*} \otimes V$ and $\overline{\epsilon_V}: V \otimes \overline{V^*} \to \id_{[m]}$  obeying the snake equations:
\begin{align}\label{eq:gendual}
\includegraphics[valign=c]{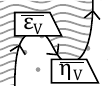}
~~=~~
\includegraphics[valign=c]{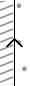}
~~=~~
\includegraphics[valign=c]{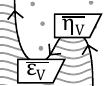}
&&
\includegraphics[valign=c]{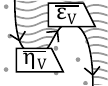}
~~=~~
\includegraphics[valign=c]{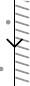}
~~=~~
\includegraphics[valign=c]{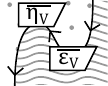}
\end{align}
We are particularly interested in duals which are \emph{standard}~\cite{Longo1997,Giorgetti2019}. Let $f \in \End(V)$ be some 1-morphism. We define the following elements $f_L \in \End(\id_{[m]})$ and $f_R \in \End(\id_{[n]})$:
\begin{align}\label{eq:standard}
\includegraphics{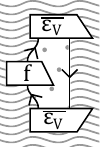}
&&
\includegraphics{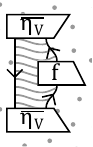}
\end{align}
Since $\End(\id_{[m]})$ and $\End(\id_{[n]})$ are commutative f.d. $C^*$-algebras, they possess a canonical trace which takes the central idempotents to $1$; we write these traces as $\Tr_{[m]}: \End(\id_{[m]}) \to \mathbb{C}$ and $\Tr_{[n]}: \End(\id_{[n]}) \to \mathbb{C}$.  We say that the duality is standard iff $\Tr_{[m]}(f_L) = \Tr_{[n]}(f_R)$ for all $f \in \End(V)$. In this case, we obtain a positive faithful trace on the $C^*$-algebra $\End(V)$.

The canonical dual we defined above is standard. In fact, standard duals are unique up to unitary equivalence; a dual $\overline{V^*}$ is standard precisely when there exists a unitary 2-morphism $U: V^* \to \overline{V^*}$ from the canonical dual such that:
\begin{align*}
\overline{\eta_V} = (U \otimes \mathbbm{1}_V) \circ \eta_V 
&&
\overline{\epsilon_V} = \epsilon_V \circ (\mathbbm{1}_V \otimes U^{\dagger})
\end{align*}

\subsection{Stinespring's theorem}

In quantum information theory, channels are identified with completely positive trace-preserving linear maps between $C^*$-algebras. In this paper we restrict ourselves to finite-dimensional (f.d.) $C^*$-algebras. We now give a brief summary of dilation theory in this setting. This is a special case of a more general theory which holds in an arbitrary rigid $C^*$-tensor category~\cite{Verdon2021,Chen2022}.

\paragraph{Splitting f.d. $C^*$-algebras.}
We will first show that every f.d. $C^*$-algebra can be \emph{split} as a \emph{pair of pants} algebra. It is well-known that every f.d. $C^*$-algebra is $*$-isomorphic to a multimatrix algebra $\bigoplus_{i} B(H_i)$, where $\{H_i\}$ are some finite-dimensional Hilbert spaces and the involution is the componentwise Hermitian adjoint.

We will first consider the case of a simple matrix algebra $B(H)$, and then generalise to an arbitrary multimatrix algebra. Recall the definition of the vector $\ket{\eta_H} \in H \otimes H$ from~\eqref{eq:etadef}. Consider the following linear isomorphism: 
\begin{align*}
\phi: B(H) &\to[\sim] H \otimes H
\\
M &\mapsto \sqrt{d}(M \otimes \mathbbm{1}) \ket{\eta_H}
\end{align*}
We will define a $*$-algebra structure on $H \otimes H$ so that $\phi$ is an isomorphism of $*$-algebras. The multiplication and unit of the algebra are defined as follows (where we use the unshaded graphical calculus for Hilbert spaces and linear maps from Section~\ref{sec:unshaded}):
\begin{align*}
\frac{1}{\sqrt{d}}\includegraphics[valign=c]{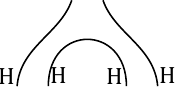}
&&
\sqrt{d}\includegraphics[valign=c]{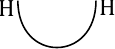}
\\
m: (H \otimes H) \otimes (H \otimes H) \to H \otimes H
&&
u: \mathbb{C} \to H \otimes H
\end{align*}
We now need a $*$-structure. For any state $\ket{\psi} \in H \otimes H$, its involution $\ket{\psi^*} \in H \otimes H$ is defined as follows:
\begin{align*}
\includegraphics[valign=c]{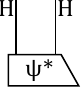}
~~:=~~
\includegraphics[valign=c]{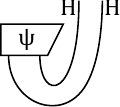}
\end{align*}
With these definitions it is very straightforward to show that $\phi$ is a unital $*$-isomorphism~\cite[Ex. 3.13]{Verdon2020b}. For obvious reasons this algebra structure on $H \otimes H$ is often called a \emph{pair of pants} algebra. 

Note that the adjoint of the unit is a linear map $u^{\dagger}: H \otimes H \to \mathbb{C}$; the composition $u^{\dagger} \circ \phi: B(H) \to \mathbb{C}$ is a trace, namely the \emph{special} trace $\overline{\Tr}:= d \, \Tr$, where $\Tr$ is the matrix trace. More generally, we define the special trace on a multimatrix algebra to be the sum of the special traces on each of the factors. We will use the special trace from now on, since it means we can directly apply results from~\cite{Verdon2021}, and it does not make any difference to the theory apart from a few scalar factors.

We now generalise to multimatrix algebras. Let $A = \bigoplus_{i=1}^m B(H_i)$ be a multimatrix algebra, where $\{H_i\}$ are some f.d. Hilbert spaces. Recall that in the shaded calculus we represent $[m]$ by wavy lines. We define a wire $H$ with the following type:
\begin{align*}
\includegraphics[valign=c]{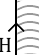}
\end{align*}
Here $H$ is an $[m]$-indexed family of Hilbert spaces, namely the Hilbert spaces $(H_i)_{i \in [m]}$. Now consider the 1-morphism $H \otimes H^*$. This is an $[m]$-indexed family $(H_i \otimes H_i)_{i \in [m]}$, where each choice of index specifies a factor of the multimatrix algebra. We define the following structure of a $*$-algebra (again called a pair of pants algebra) on $H \otimes H^*$:
\begin{align*}
\includegraphics[valign=c]{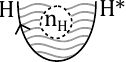}
&&
\includegraphics[valign=c]{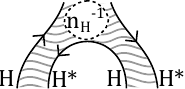}
&&
\includegraphics[valign=c]{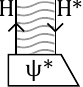}
~~=~~
\includegraphics[valign=c]{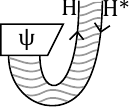}
\end{align*}
There is a $*$-isomorphism $\phi: A \cong \bigoplus_{i=1}^m B(H_i) \to H \otimes H^*$. Indeed, observe that $\End(H) = \bigoplus_{i=1}^m B(H_i)$ (where the index set for $H$ corresponds to the choice of factor). We then define:
\begin{align*}
\phi(M):= \includegraphics[valign=c]{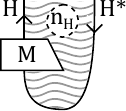}
\end{align*}
We say that either the pair of pants algebra $H \otimes H^*$, or the 1-morphism $H: 1 \to{} [m]$ itself, is a \emph{splitting} of the algebra $A$. The trace $u^{\dagger} \circ \phi$ is the special trace.

\begin{example}[Commutative $C^*$-algebras.]
We use the notation $[n]$ for the commutative $C^*$-algebra $[n]:= \bigoplus_{i \in [n]} \mathbb{C}$. (We are aware that this is the same as the notation for denoting index sets, but the context should adequately distinguish between the two uses.) Clearly this has a splitting $[n] = X \otimes X^*$, where $X = (\mathbb{C})_{i \in [n]}$.
\end{example}

\begin{example}[Splitting tensor products.]
In what follows we will often want to split the $C^*$-algebra $A \otimes B(H_1) \otimes B(H_2)$, where $A$ is some f.d. $C^*$-algebra. We will always use the splitting $(H_2 \otimes H_1 \otimes X) \otimes (X^* \otimes H_1 \otimes H_2)$, where $X: 1 \to{}[m]$ is some splitting of $A$.
\end{example}

\paragraph{Dilating channels.}
A \emph{channel} is a completely positive trace-preserving linear map. Let $A \cong \bigoplus_{i =1}^m B(H_i)$ and $B \cong \bigoplus_{j=1}^n B(K_j)$ be two multimatrix algebras. Let $H: 1 \to{}[m]$ and $K: [1] \to{} [n]$ be splittings of these algebras. By definition of the direct sum, linear maps $A \to B$ correspond precisely to 2-morphisms $H \otimes H^* \to K \otimes K^*$, which specify a linear map $B(H_i) \cong H_i \otimes H_i \to K_j \otimes K_j \cong B(K_j)$ for every choice of indices $(i,j) \in [m] \times [n]$.

We want to know when a 2-morphism $f: H \otimes H^* \to K \otimes K^*$ is completely positive and trace preserving as a linear map $A \to B$ (from now on we will simply apply these predicates to the 2-morphism). This is answered by the following theorem. 
\begin{theorem}[{Stinespring's theorem. Follows straightforwardly from~\cite[Cor. 4.13]{Selinger2007}, also proved explicitly in~\cite[Thm. 4.9]{Verdon2021}}]\label{thm:stinespring}
A 2-morphism $f: H \otimes H^* \to K \otimes K^*$ is completely positive precisely when there exists an \emph{environment} 1-morphism $E: [n] \to{} [m]$ and a \emph{dilation} 2-morphism $\tau: H \to K \otimes E$ such that the following equation is obeyed:
\begin{align*}
\includegraphics[valign=c]{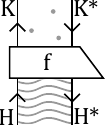}
~~=~~
\includegraphics[valign=c]{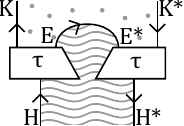}
\end{align*}
The 2-morphism $f$ is additionally trace-preserving (for the special trace) precisely when the following 2-morphism is an isometry:
\begin{align}\label{eq:tracepresisom}
\includegraphics[valign=c]{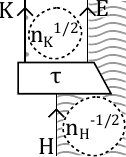}
\end{align} 

The dilation of a completely positive 2-morphism $f$ is unique up to partial isometry on the environment. That is, for any dilations $\tau_1: H \to K \otimes E_1$, $\tau_2: H \to K \otimes E_2$ of $f$, there exists a partial isometry $\alpha: E_1 \to E_2$ such that the following equations hold:
\begin{align*}
(\id_K \otimes \alpha) \circ \tau_1 = \tau_2
&&
(\id_K \otimes \alpha^{\dagger}) \circ \tau_2 = \tau_1
\end{align*} 
\end{theorem}
\begin{remark}[Minimal dilations]\label{rem:mindil}
By uniqueness of the dilation up to partial isometry, every CP morphism $f: H \otimes H^* \to K \otimes K^*$ has a \emph{minimal dilation}, unique up to a \emph{unitary} on the environment, such that the following element of the $C^*$-algebra $\End(E)$ is invertible:
\begin{align}\label{eq:mindil2morph}
\includegraphics[valign=c]{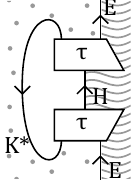}
\end{align}
We observe that this is a positive element of $\End(E)$, since it is of the form $\tilde{\tau} \circ \tilde{\tau}^{\dagger}$ for a 2-morphism $\tilde{\tau}: K^* \otimes H \to E$, where $\tilde{\tau}$ is $\tau$ with the top left leg bent down. We define $\lambda \in \End(E)$ to be the positive square root of this positive element (we will only use this notation much later on, in the proof of Theorem~\ref{thm:genpure}).
\end{remark}
\begin{remark}[Kraus maps]
The reader unaccustomed to the diagrammatic calculus might find it helpful to relate this to the description of a completely positive map in terms of Kraus operators. Let $A \cong \bigoplus_{i =1}^m B(H_i)$ and $B \cong \bigoplus_{j=1}^n B(K_j)$ be two multimatrix algebras. A completely positive map $f: A \to B$ corresponds precisely to a set of completely positive maps $f_{ij}: B(H_i) \to B(K_j)$, one for each pair of factors of $A,B$.

Let $\tau: H \to K \otimes E$ be a dilation of $f$. For each choice of indices $(i,j) \in [n] \times [m]$, let $\{\ket{v_{ijk}}\}_{k}$ be an orthonormal basis of $E_{ij}$. Then the Kraus maps of $f_{ij}: B(H_j) \to B(K_i)$ associated to this dilation and this choice of basis for $E_{ij}$ are precisely the morphisms
$$
M_{ijk} = (\mathbbm{1} \otimes \bra{v_{ijk}}) \circ \tau_{ij} : H_j \to K_i.
$$
\end{remark}

\begin{example}[Dilating states.]\label{ex:dilstates}
Channels $W: \mathbb{C} \to B(H_1) \otimes B(H_2)$ precisely correspond to states (density matrices) $\rho_W \in B(H_2 \otimes H_1) \cong B(H_1) \otimes B(H_2)$. We observe that $\mathbb{C}$ and $B(H_2 \otimes H_1)$ split as pairs of pants $\mathbb{C} \otimes \mathbb{C}$ and $(H_2 \otimes H_1) \otimes (H_1 \otimes H_2)$ respectively.

Suppose that the state is pure, i.e. $\rho_W = \ket{w} \bra{w}$ for some state $\ket{w} \in H_2 \otimes H_1$. Then the minimal dilation of $W$ has environment $E \cong \mathbb{C}$, and dilating 2-morphism $\tau = a \ket{w}: \mathbb{C} \to H_2 \otimes H_1$ for some normalising constant $a \in \mathbb{R}$. 

Let us work out the normalising constant $a$. We need~\eqref{eq:tracepresisom} to be an isometry. Since the index sets are singletons, the discs $n_H^{1/2}$ and $n_{\mathbb{C}}^{-1/2}$ are just scalars; we have $n_H^{1/2} = \dim(H_2 \otimes H_1)^{1/4} = \dim(H_2)^{1/4} \dim(H_1)^{1/4}$ and $n_{\mathbb{C}}^{-1/2}= 1$. That~\eqref{eq:tracepresisom} should be an isometry is then precisely to say that 
$$
\sqrt{\dim(H_2)\dim(H_1)} a^2 \braket{\psi | \psi} = 1
$$
which implies that $a = (\dim(H_2)\dim(H_1))^{-1/4}$.  We observe in particular that the canonical maximally entangled state of $H \otimes H$ has minimal dilation $\frac{1}{\dim(H)} \ket{\eta_H}$.
\end{example}

\section{Entanglement-invertible channels}

\subsection{Definition}
For convenience we restate the definition of entanglement-reversible and entanglement-invertible channels from the introduction.

\begin{definition}
Let $H_1,H_2$ be two Hilbert spaces, let $B(H_1)$ and $B(H_2)$ be the $C^*$-algebras of operators on these spaces and let $\sigma: B(H_1) \otimes B(H_2) \to B(H_2) \otimes B(H_1)$ be the swap channel. Let $W: \mathbb{C} \to B(H_1) \otimes B(H_2)$ be any channel (i.e. any state of $B(H_1) \otimes B(H_2)$). 

Let $M: A \otimes B(H_1) \to B$ be a channel. We say that $M$ is \emph{entanglement-reversible} with respect to $W$ if there exists a channel $N: B \otimes B(H_2) \to A$ satisfying the left equation of~\eqref{eq:entreventinv}. (The diagrams are read from bottom to top.) In this case we say that $N$ is an \emph{entanglement-left inverse of $M$ w.r.t. $W$}. If the right equation of~\eqref{eq:entreventinv} is additionally satisfied we say that $M$ is \emph{entanglement-invertible} with respect to $W$, and that $N$ is an \emph{entanglement-inverse}.
\begin{align}\nonumber
\includegraphics[valign=c]{pictures/entinvchans/entinvgeneral11.pdf}
~~=~~
\includegraphics[valign=c]{pictures/entinvchans/entinvgeneral12.pdf}
&&
\includegraphics[valign=c]{pictures/entinvchans/entinvgeneral21.pdf}
~~=~~
\includegraphics[valign=c]{pictures/entinvchans/entinvgeneral22.pdf}
\\\label{eq:entreventinv}
N \circ (M \otimes \id_{B(H_2)}) \circ  (\id_A \otimes W)
= 
\id_A
&&
M \circ (N \otimes \id_{B(H_1)}) \circ (\id_B \otimes \sigma) \circ (\id_B \otimes W)
=
\id_B
\end{align}
\end{definition}

\subsection{Quantum bijections}\label{sec:qbij}

We will begin by considering an important special case: channels which are entanglement-invertible w.r.t. the canonical maximally entangled pure state. 

\subsubsection{Characterisation in terms of minimal dilation}
We will first characterise these channels in terms of their minimal dilation. 
\begin{proposition}\label{prop:qbijmindil}
Let $M: A \otimes B(H) \to B$ be a channel, and let $\tau: H \otimes X \to Y \otimes E_{\tau}$ be a minimal dilation of $M$. Then $(M,H): A \to B$ is entanglement-invertible w.r.t. the canonical maximally entangled state of $H \otimes H$ precisely when the following 2-morphisms are unitary:
\begin{calign}\label{eq:qbijbiueqs}
\frac{1}{\dim(H)^{1/4}}~
\includegraphics[valign=c]{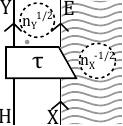}
&&
\frac{1}{\dim(H)^{1/4}}~
\includegraphics[valign=c]{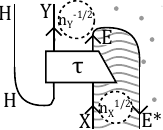}
\end{calign}
Moreover, the entanglement-inverse $(N,H): B \to A$ is uniquely determined, with the following minimal dilation:
\begin{calign}\label{eq:entinvmindil}
\includegraphics[valign=c]{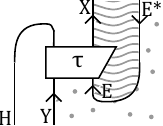}
\end{calign}
(Recall here that $n_X$ and $n_Y$ are defined as the positive square roots of the left dimensions of $X$ and $Y$ respectively.)
\end{proposition}
\begin{remark}
The 2-morphism on the left of~\eqref{eq:qbijbiueqs} is a normalised version of $\tau$; the 2-morphism on the right of~\eqref{eq:qbijbiueqs} is a normalised version of the partial transpose of $\tau$. Unitarity of a 2-morphism and its partial transpose is known as \emph{biunitarity} (see e.g.~\cite{Jones1999,Reutter2019}). A simple way to state Proposition~\ref{prop:qbijmindil} is therefore to say that a channel is entanglement-invertible w.r.t. the maximally entangled pure state precisely when its minimal dilation is biunitary (up to normalisation).
\end{remark}
\begin{proof} The following proof is partly due to D. Reutter. Let us suppose that $(M,H): A \to B$ is entanglement-invertible, and let $(N,H): B \to A$ be the entanglement-inverse. Let $\tau: H \otimes X \to Y \otimes E_{\tau}$ and $\sigma: H \otimes Y \to X \otimes E_{\sigma}$ be minimal dilations of $M$ and $N$ respectively. Then the entanglement-invertibility equations~\eqref{eq:entreventinv} reduce to the following equations for the dilations $\tau$ and $\sigma$: 
\begin{calign}\label{eq:dilentinveqs}
\frac{1}{\dim(H)}
\includegraphics[valign=c]{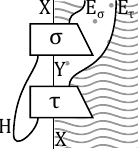}
~~=~~
\includegraphics[valign=c]{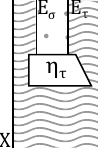}
&&
\frac{1}{\dim(H)}
\includegraphics[valign=c]{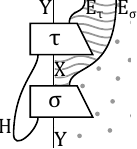}
~~=~~
\includegraphics[valign=c]{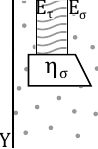}
\end{calign}
Here $\eta_{\tau}: \id_{[m]} \to E_{\sigma} \otimes E_{\tau}$ and $\eta_{\sigma}: \id_{[n]} \to E_{\tau} \otimes E_{\sigma}$ are some isometries. Let us explain how these equations were obtained. The first equation of~\eqref{eq:dilentinveqs} corresponds to the first equation of~\eqref{eq:entreventinv}. Indeed, the LHS of the first equation of~\eqref{eq:dilentinveqs} is simply a dilation of the LHS of the first equation of~\eqref{eq:entreventinv}, where we have used the fact (Example~\ref{ex:dilstates}) that the minimal dilation of the canonical maximally entangled state is $\frac{1}{\dim(H)}\ket{\eta_H}$. On the other hand, the RHS of the first equation of~\eqref{eq:dilentinveqs} is the general form for a dilation of the identity channel on $A$. Indeed, the minimal dilation of the identity channel on $A$ has trivial environment $\id_{[m]}$ and trivial 2-morphism $\id_X: X \to X$, and every other dilation is related to the minimal dilation by an isometry on the environment. The second equation of~\eqref{eq:dilentinveqs} corresponds to the second equation of~\eqref{eq:entreventinv} by the same argument. 

Now we observe the following equation:
\begin{calign}\label{eq:sigmatotaudag}
\includegraphics[valign=c]{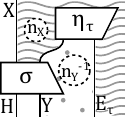}
~~=~~
\frac{1}{\dim(H)}
\includegraphics[valign=c]{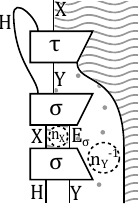}
~~=~~
\frac{1}{\sqrt{\dim(H)}}~
\includegraphics[valign=c]{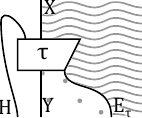}
\end{calign}
Here the first equality uses the dagger of the first equation of~\eqref{eq:dilentinveqs} (that is, reflect both the diagrams in that equation in a horizontal axis); the second equality uses the trace-preservation condition~\eqref{eq:tracepresisom} for $\sigma$. The following equation may be proven in the same way:
\begin{calign}\label{eq:tautosigmadag}
\includegraphics[valign=c]{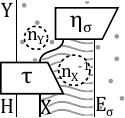}
~~=~~
\frac{1}{\sqrt{\dim(H)}}~
\includegraphics[valign=c]{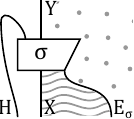}
\end{calign} 
We then obtain the following equation:
\begin{calign}
\includegraphics[valign=c]{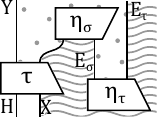}
~~=~~
\frac{1}{\sqrt{\dim(H)}}~
\includegraphics[valign=c]{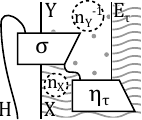}
~~=~~
\frac{1}{\dim(H)}~
\includegraphics[valign=c]{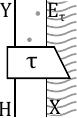}
\end{calign}
Here the first equality uses~\eqref{eq:tautosigmadag}, and the second equality uses the dagger of~\eqref{eq:sigmatotaudag}. A similar equation may be proven for $\sigma$. Since the dilations $\tau$ and $\sigma$ are minimal, we recall from Remark~\ref{rem:mindil} that $\tilde{\tau}$ and $\tilde{\sigma}$ are right-invertible, so this yields the snake equations~\eqref{eq:gendual}:
\begin{calign}\label{eq:mindilsnakes}
\includegraphics[valign=c]{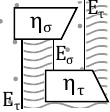}
~~=~~
\frac{1}{\dim(H)}~
\includegraphics[valign=c]{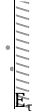}
&&
\includegraphics[valign=c]{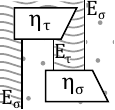}
~~=~~
\frac{1}{\dim(H)}~
\includegraphics[valign=c]{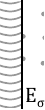}
\end{calign}
It follows that $E_{\sigma}$ is a dual for $E_{\tau}$. We now draw the wires with an upwards pointing arrow for $E_{\tau}$ and a downwards pointing arrow for $E_{\sigma}$. By~\eqref{eq:mindilsnakes}, the following morphisms are the cup and cap of a duality:
\begin{align}\label{eq:mindildualcupcap}
\sqrt{\dim(H)}~
\includegraphics[valign=c]{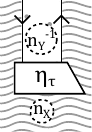}
&&
\sqrt{\dim(H)}~
\includegraphics[valign=c]{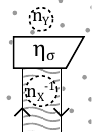}
\end{align} 
We will show that this duality is standard (this was defined in the paragraph following~\eqref{eq:standard}). We first observe the following equation:
\begin{calign}\label{eq:standardlem}
\includegraphics[valign=c]{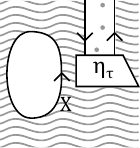}
~~=~~
\frac{1}{\dim(H)}~
\includegraphics[valign=c]{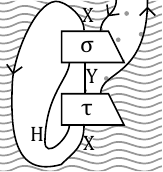}
\\
=~~
\frac{1}{\dim(H)}~
\includegraphics[valign=c]{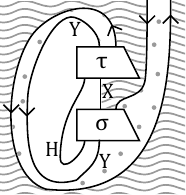}
~~=~~
\includegraphics[valign=c]{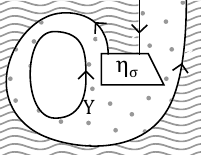}
\end{calign}
Here the first equality is by the first equation of~\eqref{eq:dilentinveqs}; the second equality is by isotopy of the diagram (pulling $\tau$ around the $X$-wire); and the third equality is by the second equation of~\eqref{eq:dilentinveqs}. 
Now standardness of the duality~\eqref{eq:mindildualcupcap} is seen as follows. Let $a \in \End(E_{\tau})$ be any morphism. Then:
\begin{align*}
\dim(H)\Tr_{[m]}\left(
\includegraphics[valign=c]{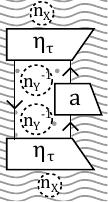}\right)
~~=~~
\dim(H)\Tr_{[m]}\left(
\includegraphics[valign=c]{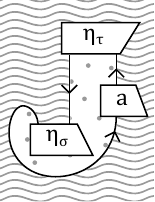}\right)
~~=~~
\Tr_{[m]}\left(
\includegraphics[valign=c]{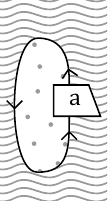}\right)
\\
=~~
\Tr_{[n]}\left(
\includegraphics[valign=c]{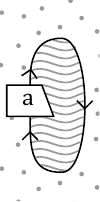}\right)
~~=~~
\dim(H)\Tr_{[n]}\left(
\includegraphics[valign=c]{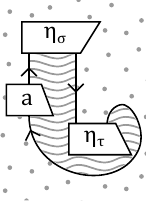}\right)
~~=~~
\dim(H)\Tr_{[n]}\left(
\includegraphics[valign=c]{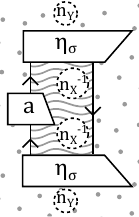}\right)
\end{align*}
Here the first equality is by~\eqref{eq:standardlem} (recall that $n_X^2=d_X$ and $(n_Y)^{-2} = (d_Y)^{-1}$); the second equality is by~\eqref{eq:mindilsnakes}; the third equality is by standardness of the canonical duality; the fourth equality is by~\eqref{eq:mindilsnakes}; and the final equality is by~\eqref{eq:standardlem} and a snake equation for the canonical duality. Since standard duals are related to the canonical dual by a unitary isomorphism, and minimal dilations are defined up to a unitary on the environment, we can identify $E_{\sigma}$ with the canonical dual of $E_{\tau}$, and the cup and cap~\eqref{eq:mindildualcupcap} with the canonical cup and cap. From~\eqref{eq:sigmatotaudag} we thereby obtain the following equation: 
\begin{calign}\label{eq:sigmabiudef}
\includegraphics[valign=c]{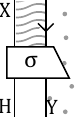}
~~=~~
\includegraphics[valign=c]{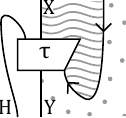}
\end{calign}
We have therefore shown that the entanglement-inverse has minimal dilation~\eqref{eq:entinvmindil}. We will now show that the first 2-morphism of~\eqref{eq:qbijbiueqs} is unitary. 
We know that it is an isometry by Theorem~\ref{thm:stinespring}, because $M$ is trace-preserving. To see that it is a coisometry:
\begin{calign}
\frac{1}{\sqrt{\dim(H)}}
\includegraphics[valign=c]{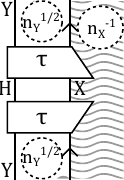}
~~=~~
\frac{1}{\sqrt{\dim(H)}}
\includegraphics[valign=c]{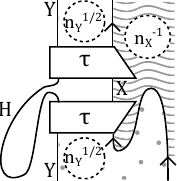}
\\
~~=~~
\frac{1}{\sqrt{\dim(H)}}
\includegraphics[valign=c]{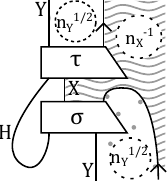}
~~=~~
\includegraphics[valign=c]{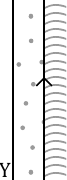}
\end{calign}
Here the first equality is by isotopy of the diagram; the second equality is by~\eqref{eq:sigmabiudef}; and the third equality is by~\eqref{eq:entinvmindil} and the fact that~\eqref{eq:mindildualcupcap} are standard. 

Unitarity of the second 2-morphism of~\eqref{eq:qbijbiueqs} follows immediately by symmetry of the entanglement-invertibility equations in $\tau$ and $\sigma$. 

We now need only prove the other direction: if the 2-morphisms~\eqref{eq:qbijbiueqs} are unitary, then $(M,H): A \to B$ is entanglement-invertible. We claim that the dilation~\eqref{eq:entinvmindil} specifies an entanglement-inverse. By Theorem~\ref{thm:stinespring} it indeed dilates a channel, since the right-hand 2-morphism of~\eqref{eq:qbijbiueqs} is a coisometry. The entanglement-invertibility equations~\eqref{eq:dilentinveqs} are then seen as follows:
\begin{align*}
\frac{1}{\dim(H)}~
\includegraphics[valign=c]{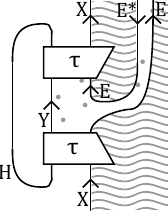}
~~=~~
\frac{1}{\sqrt{\dim(H)}}~
\includegraphics[valign=c]{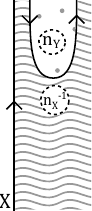}
\\
\frac{1}{\dim(H)}~
\includegraphics[valign=c]{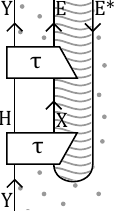}~~=~~
\frac{1}{\sqrt{\dim(H)}}~
\includegraphics[valign=c]{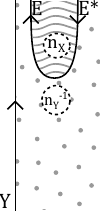}
\end{align*}
Here the first equation is by the fact that the second 2-morphism of~\eqref{eq:qbijbiueqs} is an isometry. The second equation is by the fact that the first 2-morphism of~\eqref{eq:qbijbiueqs} is a coisometry.
\end{proof}

\subsubsection{Compositional structure}
We will now show that these channels entanglement-invertible w.r.t. the canonical maximally entangled state are precisely the quantum bijections which were previously studied in the setting of noncommutative combinatorics~\cite{Musto2018}. We can then directly apply results about their compositional structure from that work.
\begin{definition}[{\cite[Def. 4.3]{Musto2018}}]\label{def:qbij}
Let $A \cong X \otimes X^*$ and $B \cong Y \otimes Y^*$ be f.d. $C^*$-algebras, and let $H$ be a Hilbert space. A \emph{quantum bijection} $(M,H): A \to B$ is a channel $A \otimes B(H) \to B$ whose minimal dilation $\tau: H \otimes X \to Y \otimes E$ obeys the following  additional equations:
\begin{align}\label{eq:qbijmult}
\frac{1}{\dim(H)}
\includegraphics[valign=c]{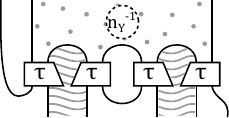}
~~=~~
\frac{1}{\sqrt{\dim(H)}}
\includegraphics[valign=c]{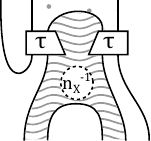}
\end{align}
\begin{align}\label{eq:qbijcomult}
\frac{1}{\dim(H)}
\includegraphics[valign=c]{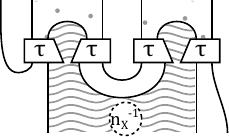}
~~=~~
\frac{1}{\sqrt{\dim(H)}}
\includegraphics[valign=c]{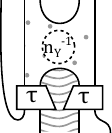}
\end{align}
\begin{align}\label{eq:qbijunit}
\frac{1}{\sqrt{\dim(H)}}
\includegraphics[valign=c]{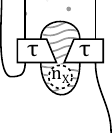}
~~=~~
\includegraphics[valign=c]{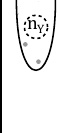}
\end{align}
\end{definition}
\begin{remark}
Definition~\ref{def:qbij} is more concise than~\cite[Def. 4.3]{Musto2018}, which had five equations; the two omitted equations are implied by the statement that $(M,H)$ is a channel.
\end{remark}
\begin{lemma}\label{lem:qbijentinv}
Entanglement-invertible channels $(M,H): A \to B$ are precisely quantum bijections.
\end{lemma}
\begin{proof}
Suppose that $(M,H)$ is an entanglement-invertible channel, and therefore the 2-morphisms~\eqref{eq:qbijbiueqs} are unitary. For~\eqref{eq:qbijmult}:
\begin{align*}
\frac{1}{\dim(H)}
\includegraphics[valign=c]{pictures/entinvchans/qbijpfmult1.pdf}
~~=~~
\frac{1}{\dim(H)}
\includegraphics[valign=c]{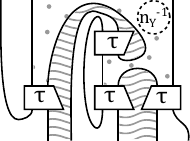}
~~=~~
\frac{1}{\sqrt{\dim(H)}}
\includegraphics[valign=c]{pictures/entinvchans/qbijpfmult4.pdf}
\end{align*}
Here the first equality is by isotopy of the diagram; the second equality is by the fact that the second 2-morphism of~\eqref{eq:qbijbiueqs} is an isometry.
The equation~\eqref{eq:qbijcomult} is immediate from the fact that the first 2-morphism of~\eqref{eq:qbijbiueqs} is a coisometry. For~\eqref{eq:qbijunit}:
\begin{align*}
\frac{1}{\sqrt{\dim(H)}}
\includegraphics[valign=c]{pictures/entinvchans/qbijpfunit1.pdf}
~~=~~
\frac{1}{\sqrt{\dim(H)}}
\includegraphics[valign=c]{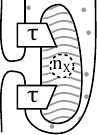}~~=~~
\includegraphics[valign=c]{pictures/entinvchans/qbijpfunit4.pdf}
\end{align*}
Here the first equality is by isotopy of the diagram, and the second equality is by the fact that the second 2-morphism of~\eqref{eq:qbijbiueqs} is a coisometry. 

In the other direction, suppose that $(M,H): A \to B$ is a quantum bijection, and let $\tau: H \otimes X \to Y \otimes E$ be a minimal dilation of $M$. We know that the first 2-morphism of~\eqref{eq:qbijbiueqs} is an isometry, since $\tau$ is a channel. The other three biunitarity equations are shown by a process which is essentially the inverse of the first half of this proof. For example,~\eqref{eq:qbijmult} implies that the second 2-morphism of~\eqref{eq:qbijbiueqs} is an isometry:
\begin{align*}
\frac{1}{\dim(H)}
\includegraphics[valign=c]{pictures/entinvchans/qbijpfmult1.pdf}
~~&=~~
\frac{1}{\sqrt{\dim(H)}}
\includegraphics[valign=c]{pictures/entinvchans/qbijpfmult4.pdf}
\\
\Leftrightarrow \qquad
\frac{1}{\dim(H)}~
\includegraphics[valign=c]{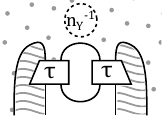}
~~&=~~
\frac{1}{\sqrt{\dim(H)}}~
\includegraphics[valign=c]{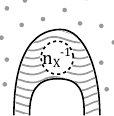}
\\
\Leftrightarrow \qquad
\frac{1}{\sqrt{\dim(H)}}
\includegraphics[valign=c]{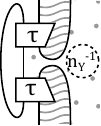}
~~&=~~
\includegraphics[valign=c]{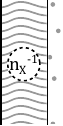}
\end{align*}
Here the first implication is by bending the output wires down and using minimality of the dilation, which implies right invertibility of $\tilde{\tau}$ (Remark~\ref{rem:mindil}); the second implication is by bending the two leftmost input wires upwards and isotopy of the diagram. The last equation is clearly the isometry condition for the second 2-morphism of~\eqref{eq:qbijbiueqs}. The other two biunitarity equations are shown similarly.
\end{proof}
\noindent
We can therefore apply the compositional framework developed in~\cite{Musto2018} to the study of these entanglement-invertible channels. We showed in that work that quantum bijections properly form a 2-category $\QBij$ whose objects are f.d. $C^*$-algebras, whose 1-morphisms are quantum bijections, and whose morphisms are \emph{intertwiners}; moreover, the relationship between a quantum bijection and its entanglement-inverse is one of 2-categorical duality. Here we will highlight two facts.
\begin{itemize}
\item Let $(M_1,H_1), (M_2,H_2): A \to B$ be quantum bijections, with minimal dilations $\tau_1: H_1 \otimes X \to Y \otimes E_1$ and $\tau_2: H_2 \otimes X \to Y \otimes E_2$ respectively. We define an \emph{intertwiner} $f: (M_1,H_1) \to (M_2,H_2)$ to be a linear map $f: H_1 \to H_2$ satisfying the following equation:
\begin{align}\label{eq:intertwinerdef}
\includegraphics[valign=c]{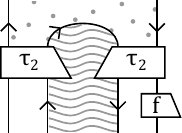}
~~=~~
\includegraphics[valign=c]{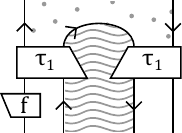}
\end{align}
Quantum bijections $A \to B$ are the objects of a category $\QBij(A,B)$, whose morphisms are these intertwiners. We say that two quantum bijections are \emph{isomorphic} if they are related by a unitary intertwiner.

\item Let $(M_1,H_1), (M_2,H_2): A \to B$ be quantum bijections. The \emph{direct sum} $(M_1 \oplus M_2, H_1 \oplus H_2): A \to B$ is the quantum bijection whose defining channel is $M_1 \oplus M_2: A \otimes B(H_1 \oplus H_2) \to B$. We say that a quantum bijection is \emph{simple} if it cannot be decomposed as a nontrivial direct sum. We showed in~\cite[Thm. 6.4]{Musto2018} that every quantum  bijection is isomorphic to a finite direct sum of simple quantum bijections.
\end{itemize}
The following lemma will be useful later on.
\begin{lemma}\label{lem:qbijexists}
Let $A$ and $B$ be f.d. $C^*$-algebras. There exists a quantum bijection $A \to B$ precisely when $\dim(A) = \dim(B)$. 
\end{lemma}
\begin{proof}
That $\dim(A) = \dim(B)$ if there exists a quantum bijection $A \to B$ was shown in~\cite[Thm. 4.8]{Musto2018}. We will show the other direction now. Let $D = \dim(A) = \dim(B)$. We know that $A$ and $B$ are multimatrix algebras, i.e. $A \cong \bigoplus_{i \in [m]} B(H_i)$ and $B \cong \bigoplus_{j \in [n]} B(K_j)$. The composition of two quantum bijections is a quantum bijection, so it is sufficient to define entanglement-invertible channels $A \to{} [D]$ and $[D] \to{} B$. 

We now describe how to construct the first quantum bijection $A \to{} [D]$. We first observe that if $A$ is a matrix algebra, then the tight teleportation scheme of~\cite{Werner2001} is already a quantum bijection $A \to{} [D]$. To extend this to the case of a multimatrix algebra, let $\mu$ be the lowest common multiple of all the $\{\dim(H_i)\}_{i \in [m]}$. Let $H$ be a Hilbert space of dimension $\mu$. The quantum bijection is defined as follows: first perform a projective measurement onto the factors of $A$, which will produce an outcome $i \in [m]$; then perform the direct sum of $\mu/\dim(H_i)$ copies of a tight teleportation scheme $B(H_i) \to{} [\dim(H_i)^2]$. 

A quantum bijection $[D] \to{} B$ may be constructed similarly. 
\end{proof}
\noindent
Finally, we note that, since the category $\QBij(A,B)$ has a semisimple structure, one might expect it to be the category of representations of some algebraic object. This is indeed the case; $\QBij(A,B)$ is the category of f.d. $*$-representations of a \emph{Hopf-Galois object} for the \emph{quantum permutation group} of $A$~\cite{Musto2019, Verdon2020}.

\subsection{General entanglement-reversible and entanglement-invertible channels}

Having considered channels entanglement-invertible w.r.t. the canonical maximally entangled state in some detail, we now turn our attention to general entanglement-reversible and entanglement-invertible channels. In Section~\ref{sec:generalclassification} we will characterise these channels in terms of their minimal dilations, while in Section~\ref{sec:wernerex} we will show how this generalises Werner's classification of tight teleportation and dense coding protocols in terms of unitary error bases. 

\subsubsection{Characterisation in terms of minimal dilation}\label{sec:generalclassification}

We will now answer the question: given a channel $(M,H_1): A \to B$ and a state $W: \mathbb{C} \to B(H_1) \otimes B(H_2)$, when is the channel $M$ entanglement-reversible/entanglement-invertible w.r.t. $W$? 

For clarity, we will split the result into two parts. In Theorem~\ref{thm:genpure} we will assume that $W$ is pure.  Then, in Corollary~\ref{cor:genmixed}, we will extend the result to mixed $W$. 

\paragraph{Result for pure states.}  As discussed in Example~\ref{ex:dilstates}, for any pure state $W: \mathbb{C} \to B(H_1) \otimes B(H_2)$ there exists some state $\ket{w} \in H_2 \otimes H_1$ such that $W$ has minimal dilation $(\dim(H_2)\dim(H_1))^{-1/4} \ket{w}: \mathbb{C} \to H_2 \otimes H_1$. There is a uniquely defined linear map $\omega: H_1 \to H_2$  such that $(\dim(H_2)\dim(H_1))^{-1/4} \ket{w} = (\omega \otimes \mathbbm{1}_{H_1}) \ket{\eta_{H_1}}$, where $\ket{\eta_{H_1}}: \mathbb{C} \to H_1 \otimes H_1$ is defined as in~\eqref{eq:etadef}. This yields a bijective correspondence between pure states and such linear maps. We will from now on refer to $W: \mathbb{C} \to B(H_2) \otimes B(H_1)$ as `the pure state defined by $\omega: H_1 \to H_2$'.

The following lemma will allow us to reduce to the case where $\omega$ is invertible, at least when $W$ is pure. We first define some notation. A general $\omega$ can obviously be decomposed as $\omega = i_{\omega} \circ \bar{\omega} \circ q_{\omega}$, where $i_{\omega}: \Image(\omega) \to H_2$ is an isometry, $q_{\omega}: H_1 \to H_1/\Ker({\omega})$ is a coisometry, and $\bar{\omega}: H_1/\Ker({\omega}) \to \Image(\omega)$ is an isomorphism.  Let $M: A \otimes B(H_1) \to B$ be a channel, and let $\tau: H_1 \otimes X \to Y \otimes E$ be the minimal dilation. We define a channel $\bar{M}: A \otimes B(H_1/\Ker(\omega)) \to B$ whose dilation is a scalar multiple of $\tau \circ (q_{\omega}^T \otimes \mathbbm{1}_{X}): H_1/\Ker(\omega) \otimes X \to Y \otimes E$ (where the scalar multiple is chosen so that the dilation satisfies the trace-preservation condition~\eqref{eq:tracepresisom}).
\ignore{$\frac{(\dim(H_1))^{1/4}}{(\dim(H_1/\Ker(\omega)))^{1/4}} \tau \circ (q_{\omega}^T \otimes \mathbbm{1}_{X}): H_1/\Ker(\omega) \otimes X \to Y \otimes E$ (it may readily checked that this dilation satisfies the trace-preservation condition~\eqref{}).}
Finally, we define $\bar{W}: \mathbb{C} \to  B(H_1/ \Ker(\omega)) \otimes B(\Image(\omega))$ to be the pure state defined by $\bar{\omega}$.
\begin{lemma}\label{lem:entrevquotient}
The channel $M$ is entanglement-reversible/entanglement-invertible w.r.t $W$ precisely when $\bar{M}$ is entanglement-reversible/entanglement-invertible w.r.t. $\bar{W}$.
\end{lemma}
\begin{proof}
Suppose that $(M,H_1): A \to B$ is entanglement-reversible/entanglement-invertible w.r.t. $W$. Let $(N,H_2): B \to A$ be the entanglement-left inverse/entanglement-inverse, and let $\sigma: H_2 \otimes Y \to X \otimes E_{\sigma}$ be a minimal dilation of $N$. 
As discussed in the proof of Proposition~\ref{prop:qbijmindil}, in terms of the dilations, the entanglement-reversibility/entanglement-invertibility equations~\eqref{eq:entreventinv} are as follows:
\begin{align}\nonumber
\includegraphics[valign=c]{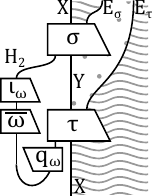}
~~=~~
\includegraphics[valign=c]{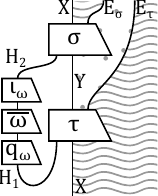}
~~=~~
\includegraphics[valign=c]{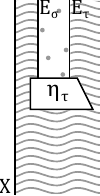}
\\\label{eq:entrevinvquotient}
\includegraphics[valign=c]{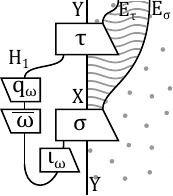}
~~=~~
\includegraphics[valign=c]{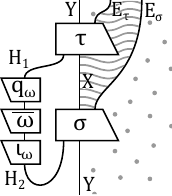}
~~=~~
\includegraphics[valign=c]{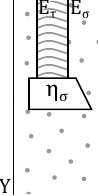}
\end{align}
Here the first equalities are by isotopy of the diagram. Let us define a channel $\bar{N}: B \otimes B(\Image(\omega)) \to A$ whose dilation is a scalar multiple of $\sigma \circ (\iota_{\omega} \otimes \mathbbm{1}_Y): \Image(\omega) \otimes Y \to X \otimes E$ (where, again, the scalar multiple is chosen so that the trace-preservation condition~\eqref{eq:tracepresisom} is satisfied). But now the equations~\eqref{eq:entrevinvquotient} (where we consider the first and third terms in each equation) precisely state that $\bar{N}$ is an entanglement-left inverse/entanglement-inverse for $\bar{M}$ w.r.t. $\bar{W}$.

On the other hand, suppose that $\bar{M}$ is entanglement-reversible/entanglement-invertible w.r.t $\bar{W}$. Then there is a channel $\bar{N}: B \otimes B(\Image(\omega)) \to A$ which is an entanglement-left inverse/entanglement-inverse of $\bar{M}$. But $B \otimes B(\Image(\omega))$ is a unital $*$-subalgebra of $B \otimes B(H_2)$ by the isometry $\iota_{\omega}: \Image(\omega) \to H_2$; so by Arveson's extension theorem~\cite[Thm. 1.2.3]{Arveson1969} there is a (non-unique) extension $N: B \otimes B(H_2) \to A$. Let $\sigma: H_2 \otimes Y \to X \otimes E$ be a minimal dilation of $N$; then the fact that $N$ is an extension of $\bar{N}$ with respect to the isometry $\iota_{\omega}$ implies the relevant equations~\eqref{eq:entrevinvquotient}.
\end{proof}
\noindent
Lemma~\ref{lem:entrevquotient} implies that, at least in the case where $W$ is pure, we can reduce to the case where $\omega$ is invertible. In this case, we identify $H_1 = H_2 = H$. This is the context for the following theorem. 
\begin{theorem}\label{thm:genpure}
Let $H$ be an f.d. Hilbert space, and let $W: \mathbb{C} \to B(H) \otimes B(H)$ be the pure state defined by an invertible linear map $\omega: H \to H$. Let $A$ and $B$ be any f.d. $C^*$-algebras, and let $X: [1] \to{} [m]$ and $Y: [1] \to{} [n]$ be splittings of $A$ and $B$ respectively.

Let $M: A \otimes B(H) \to B$ be a channel, and let $\tau: H \otimes X \to Y \otimes E$ be a minimal dilation of $M$. Then:
\begin{enumerate}
\item The channel $(M,H)$ is entanglement-reversible with respect to $W$ precisely when there exists an positive invertible element $\kappa \in \End(E^*)$ such that the following 2-morphisms are isometries:
\begin{align}\label{eq:biisomeqs}
\frac{1}{\dim(H)^{1/4}}
\includegraphics[valign=c]{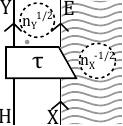}
&&
\frac{1}{\dim(H)^{1/4}}
\includegraphics[valign=c]{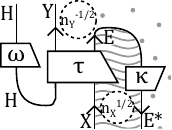}
\end{align}
\item Suppose that $(M,H)$ is entanglement-reversible with respect to $W$. Then $\dim(A) \leq \dim(B)$. The isometries~\eqref{eq:biisomeqs} are unitary precisely when $\dim(A) = \dim(B)$; in this case the entanglement-left inverse $N: B \otimes B(H) \to A$ is uniquely defined, with the following minimal dilation:
\begin{align}\label{eq:entinvgenmindil}
\includegraphics[valign=c]{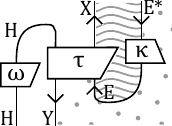}
\end{align}
\item The channel $(M,H)$ is entanglement-invertible with respect to $W$ precisely when the following conditions are satisfied:
\begin{itemize}
\item $(M,H)$ is a quantum bijection.
\item The linear map $\omega^{\dagger} \circ \omega :H \to H$ is an intertwiner $(M,H) \to (M,H)$.
\end{itemize}
\end{enumerate}
\end{theorem}
\begin{proof}
We prove each statement in turn. 

\vspace{.2cm}
\noindent
\emph{Proof of 1.} The first 2-morphism of~\eqref{eq:biisomeqs} is always an isometry by Theorem~\ref{thm:stinespring}, since the channel is trace-preserving. We therefore need to prove that $(M,H)$ is entanglement-reversible iff the other 2-morphism in~\eqref{eq:biisomeqs} is an isometry.

By Remark~\ref{rem:mindil}, since the dilation $\tau$ is minimal, there exists a morphism $\overline{\tau}: H \otimes X \to Y \otimes E$ and a positive invertible morphism $\lambda: E \to E$ such that the following equations hold:
\begin{align*}
\includegraphics[valign=c]{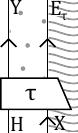}
~~=~~
\includegraphics[valign=c]{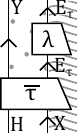}
&&
\includegraphics[valign=c]{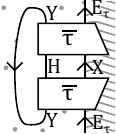}
~~=~~
\includegraphics[valign=c]{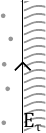}
\end{align*}
We define the following positive element $T \in \End(Y^* \otimes H \otimes X)$:
\begin{align}\label{eq:bigtdef}
\includegraphics[valign=c]{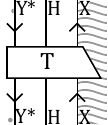}
~~:=~~
\includegraphics[valign=c]{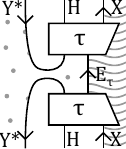}
~~=~~
\includegraphics[valign=c]{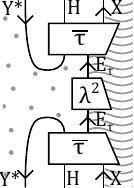}
\end{align}
These 2-morphisms are indexed families of linear maps; we now choose some indices. Let $i \in [n]$, $j \in [m]$ be the indices of the left and right shaded regions respectively in~\eqref{eq:bigtdef}. Let $E= (E_{ij})_{(i,j) \in [n] \times [m]}$, $Y = (Y_{i})_{i \in [n]}$ and  $X = (X_{j})_{j \in [m]}$. Choose some orthonormal basis $\{\ket{k}\}_{k \in [\dim(E_{ij})]}$ for $E_{ij}$ in which $\lambda_{ij} \in \End(E_{ij})$ is diagonal. Let $\lambda_{ijk}:= \bra{k} \lambda_{ij} \ket{k}$, and let $\overline{\tau}_{ijk}:= (\id_{Y} \otimes \bra{k}) \circ \overline{\tau}_{ij}$. Then for each $i,j$ we can expand $T_{ij}$ as follows:
\begin{align}\label{eq:bigtexpand}
\includegraphics[valign=c]{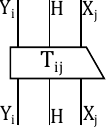}
~~=
\sum_{k \in [\dim(E_{ij})]} \lambda_{ijk}^2 \includegraphics[valign=c]{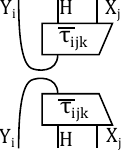}
\end{align}
\ignore{
Because $\tau$ is an isometry we have the following equation:
\begin{align*}
\sum_{\substack{j \in J \\ k \in K_{ij}}} \frac{\sqrt{\dim(Y_j)}}{\sqrt{\dim(H) \dim(X_i)}}
\includegraphics[valign=c]{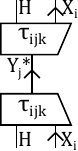}
~~=~~
\includegraphics[valign=c]{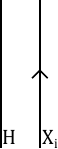}
\end{align*}}
We now use the fact that a channel is reversible iff its quantum  confusability graph is discrete. This was shown in~\cite[Thm 4.4]{Verdon2023}. Indeed, the first equation of~\eqref{eq:entreventinv} corresponds to reversibility of the following channel:
\begin{align}\label{eq:chantoreverse}
\includegraphics{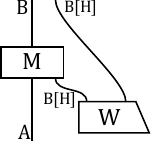}
\end{align} 
By~\cite[Def. 3.9, Prop. 3.11]{Verdon2023}, the confusability graph of the channel~\eqref{eq:chantoreverse} is discrete iff the following equation holds:
\begin{align}\label{eq:supporteq}
\supp\left(\includegraphics[valign=c]{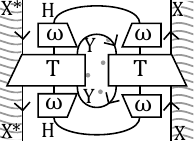}\right)
~~=~~
\includegraphics[valign=c]{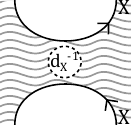}
\end{align}
The discrete graph is the minimal confusability graph of a channel; that is to say, the support on the LHS of~\eqref{eq:supporteq} can be no smaller than the projection on the RHS of~\eqref{eq:supporteq}. Therefore, since a positive element is preserved under conjugation by any projection containing its support, the equation~\eqref{eq:supporteq} is equivalent to:
\begin{align}\label{eq:supportconjeq}
\includegraphics[valign=c]{pictures/entinvchans/poseldefw.pdf}
~~=~~
\includegraphics[valign=c]{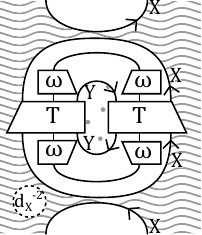}
\end{align}
Let us use the indices $i_2 \in [m]$, $j \in [n]$ and $i_1 \in [m]$ for the left, central and right shaded regions respectively of the morphism on the left hand side of the equality~\eqref{eq:supportconjeq}. Then for any choice of $i_1$, $i_2$ we obtain the following equation for the component linear maps:
\begin{align}\label{eq:supportconjcomponents}
\sum_{j\in J}~\includegraphics[valign=c]{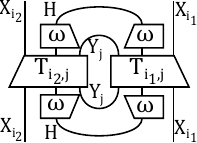}
~~=~~
\frac{\delta_{i_1,i_2}}{\dim(X_{i_1})^2}\sum_{j \in J}~\includegraphics[valign=c]{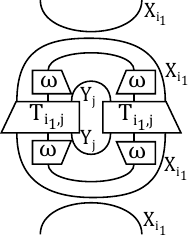}
\end{align}
Inserting~\eqref{eq:bigtexpand} in~\eqref{eq:supportconjcomponents} and using isotopy of the diagrams, we obtain the following equation:
\begin{align*}
&\sum_{\substack{j \in J \\ k_1 \in [\dim(E_{i_1j})] \\ k_2 \in [\dim(E_{i_2j})]}}\lambda^2_{i_1jk_1} \lambda^2_{i_2 j k_2}
\includegraphics[valign=c]{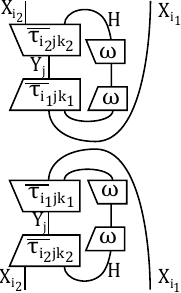}
\\&~~~~~~~~~=~~
\sum_{\substack{j \in J \\ k_1, k_2 \in [\dim(E_{i_1j})]}} 
\frac{\delta_{i_1,i_2}\lambda^2_{i_1jk_1} \lambda^2_{i_1 j k_2}}{\dim(X_{i_1})^2}
\includegraphics[valign=c]{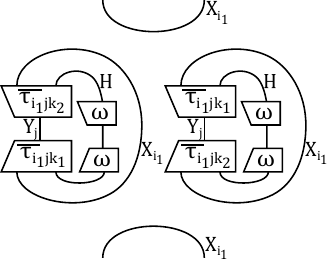}
\end{align*}
By positivity, this equation is satisfied precisely when 
\begin{align*}
\includegraphics[valign=c]{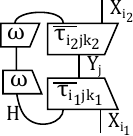}
~~=~~\delta_{i_1,i_2}\nu_{i_1jk_1k_2}
\includegraphics[valign=c]{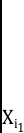}
\end{align*}
for some $\nu_{i_1 j k_1 k_2} \in \mathbb{C}$. Moving back to the shaded calculus, this is precisely to say that there is a 2-morphism $\nu: E^* \to E^*$ satisfying the following equation:
\begin{align}\label{eq:entrevpurecond}
\includegraphics[valign=c]{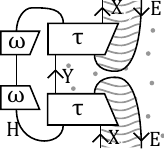}
~~=~~
\includegraphics[valign=c]{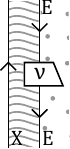}
\end{align}
If such a $\nu$ exists, we observe that it has full support. Indeed, suppose that this is not the case; then there is some projection $k \in \End(E_{\tau})$ such that $\nu \circ k^T = 0$. By~\eqref{eq:entrevpurecond} and invertibility of $\omega$ this implies that $(\mathbbm{1}_Y \otimes k) \circ \tau = 0$. But this contradicts the fact that $\tau$ is a minimal dilation. We can therefore define $\kappa:= n_X^{-1/2} \otimes \nu^{-1/2} \otimes n_Y^{1/2}$, and then the second morphism of~\eqref{eq:biisomeqs} will be an isometry.

\vspace{.2cm}
\noindent
\emph{Proof of 2.} For the convenience of the reader we restate the claim in (2): 
\begin{quote}
    Suppose that $(M,H)$ is entanglement-reversible with respect to $W$. Then $\dim(A) \leq \dim(B)$. The isometries~\eqref{eq:biisomeqs} are unitary precisely when $\dim(A) = \dim(B)$; in this case the entanglement-left inverse $N: B \otimes B(H) \to A$ is uniquely defined, with the following minimal dilation:
\begin{align}\label{eq:entinvgenmindil}
\includegraphics[valign=c]{pictures/entinvchans/entrevdef.pdf}
\end{align}
\end{quote}

We will first show that entanglement-reversibility implies $\dim(A) \leq \dim(B)$, with unitarity of the 2-morphisms~\eqref{eq:biisomeqs} iff this is an equality. For conciseness we will write $I: H \otimes X \to Y \otimes E$ for the first isometric 2-morphism in~\eqref{eq:biisomeqs}, and $I': X \otimes E^* \to H \otimes Y$ for the second isometric 2-morphism in~\eqref{eq:biisomeqs}. We have the following equations for $I$:
\begin{align*}
\includegraphics[valign=c]{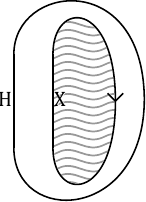}
~~=~~
\includegraphics[valign=c]{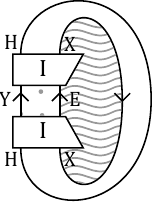}
~~=~~
\includegraphics[valign=c]{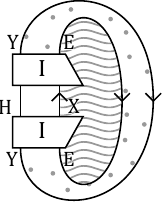}
\\
\includegraphics[valign=c]{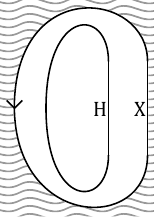}
~~=~~
\includegraphics[valign=c]{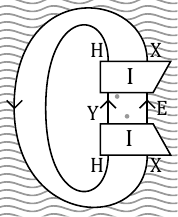}
~~=~~
\includegraphics[valign=c]{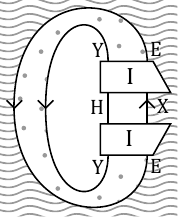}
\end{align*}
Here on each line the first equality is by the definition of an isometry and the second equality is by isotopy, pulling the $I$-box round the loop. 
Choosing indices $i \in [m], j \in [n]$ and using~\eqref{eq:tracedef}, we see that these lines reduce to: 
\begin{align*}
\sum_{i \in [m]}\dim(H) \dim(X_i) = \sum_{(i,j) \in [m] \times [n]} \Tr(I_{ij} I_{ij}^{\dagger}) &\leq \sum_{(i,j) \in [m] \times [n]} \dim(Y_j) \dim(E_{ji})
\\
\dim(H) \dim(X_i) = \sum_{j \in [n]} \Tr(I_{ij} I_{ij}^{\dagger}) &\leq \sum_{j \in [n]} \dim(Y_j) \dim(E_{ji})
\end{align*}
Here in both lines the inequality comes from $I_{ij} I_{ij}^{\dagger} \leq \mathbbm{1}_{Y_{j} \otimes E_{ji}}$, which follows from the fact that $I_{ij}$ is an isometry. We can do the same thing for the isometry $I'$. Altogether, we obtain four inequalities for $\dim(H)$:
\begin{align}\label{eq:topinequality}
\frac{\sum_{(i, j) \in [m] \times [n]} \dim(X_i) \dim(E_{ji})}{\sum_{j \in [n]}\dim(Y_j)} \leq \dim(H) \geq \frac{\sum_{i \in [m]} \dim(X_{i}) \dim(E_{ji})}{\dim(Y_j)}~~\forall~j \in [n]
\\\label{eq:bottominequality}
\frac{\sum_{(i,j) \in [m] \times [n]} \dim(Y_j) \dim(E_{ji})}{\sum_{i \in [m]}\dim(X_i)} \geq \dim(H) \leq \frac{\sum_{j \in [n]} \dim(Y_j) \dim(E_{ji})}{\dim(X_i)}~~\forall~i \in [m]
\end{align}
We will first make the assumption that $\dim(X_i) =: d_X$ and $\dim(Y_j) =: d_Y$ do not vary over $i \in [m], j \in [n]$; we will then extend this to the general result. Now, starting from the left inequality of~\eqref{eq:topinequality}:
\begin{align}\nonumber
\dim(H) \geq 
\frac{\sum_{i,j} d_X \dim(E_{ji})}{\sum_{j} d_Y} = \frac{d_X^2}{n d_Y^2} \sum_{ij} \frac{d_Y \dim(E_{ji})}{d_X} &\geq \frac{d^2_X}{nd^2_Y} \sum_{i \in I} \dim(H) 
\\\label{eq:exploringinequalities}&~~~= \frac{m d^2_X}{nd^2_Y} \dim(H) = \frac{\dim(A)}{\dim(B)} \dim(H) 
\end{align}
Here we used the right inequality of~\eqref{eq:bottominequality}. It follows immediately that $\dim(A) \leq \dim(B)$. Now by positivity and faithfulness of the standard trace~\eqref{eq:standard}, the isometric 2-morphisms~\eqref{eq:biisomeqs} are unitary iff one inequality from~\eqref{eq:topinequality} and one equality from~\eqref{eq:bottominequality} are equalities (it will then follow that all the inequalities are equalities). We have seen by~\eqref{eq:exploringinequalities} that if $\dim(A) = \dim(B)$ then the left inequality of~\eqref{eq:topinequality} is an equality; it may be shown similarly that the left inequality of~\eqref{eq:bottominequality} is an equality also, and so the 2-morphisms~\eqref{eq:biisomeqs} are unitary. On the other hand, if $\dim(A) \neq \dim(B)$, then by~\eqref{eq:exploringinequalities} either the top left or the bottom right inequality must be strict, and so at least one of the 2-morphisms~\eqref{eq:biisomeqs} is not unitary.

We now remove the assumption that $\dim(X_i)$ and $\dim(Y_j)$ do not vary over $i \in [m]$, $j \in [n]$. By Lemma~\ref{lem:qbijexists}, there exist quantum bijections $(O,H_O): [\dim(A)] \to A$ and $(P,H_P): B \to{} [\dim(B)]$. The composition $(P \circ M \circ O, H_O \otimes H \otimes H_P)$ is an entanglement-reversible channel $[\dim(A)] \to{} [\dim(B)]$; this satisfies the assumptions we made in the last paragraph since all the factors in the source and target are one-dimensional. We therefore have $\dim(A) \leq \dim(B)$. Now let $d: Y \to Z_{B} \otimes E_P$ and $e: Z_{A} \to X \otimes E_O$ be minimal dilations of $P$ and $O$ respectively. By Proposition~\ref{prop:qbijmindil}, these minimal dilations obey the equations~\eqref{eq:qbijbiueqs}. Using this fact we will show that $d \circ \tau \circ e$ is a minimal dilation for $P \circ M \circ O$ (recall the definition of minimality from Remark~\ref{rem:mindil}). In the following equation we have shaded the region corresponding to $[\dim(A)]$ with tiny dots and the region corresponding to $[\dim(B)]$ with a checkerboard pattern:
\begin{calign}
\includegraphics[valign=c]{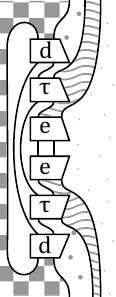}
~~=~~
\includegraphics[valign=c]{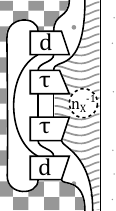}
~~=~~
\includegraphics[valign=c]{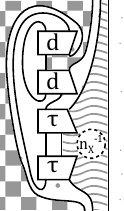}
~~=~~
\includegraphics[valign=c]{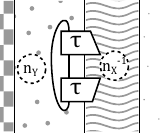}
\end{calign}
Here the for the first equality we use unitarity of the left 2-morphism of~\eqref{eq:qbijbiueqs} for $e$; the second equality is by isotopy, pulling $d^{\dagger}$ around the loop; and for the third equality we use unitarity of the right 2-morphism of~\eqref{eq:qbijbiueqs} for $d$. Clearly, the final expression is invertible, since $\tau$ is a minimal dilation; therefore the dilation $d \circ \tau \circ e$ is minimal.
The result of the last paragraph therefore applies, and the following 2-morphisms are unitary iff $\dim(A) = \dim(B)$:
\begin{align}\label{eq:compobiueqs}
\includegraphics[valign=c]{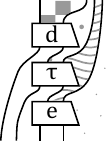}
&&
\includegraphics[valign=c]{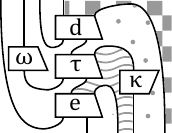}
\end{align}
But by the biunitarity equations~\eqref{eq:qbijbiueqs} for $e$ and $d$, it is clear that the 2-morphisms~\eqref{eq:compobiueqs} are unitary if and only if the 2-morphisms~\eqref{eq:biisomeqs} are. We have therefore extended the results of the last paragraph to the general case. 

Finally, we must show that the entanglement-left inverse has the minimal dilation~\eqref{eq:entinvgenmindil}. Let $\sigma: Y \otimes E_{\sigma} \to X$ be a minimal dilation of the entanglement-inverse. Then assuming entanglement-reversibility we obtain the following implication:
\begin{align}\label{eq:sigmaimplication}
\includegraphics[valign=c]{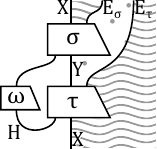}
~~=~~
\includegraphics[valign=c]{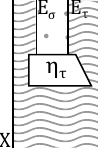}
&& \Rightarrow &&
\includegraphics[valign=c]{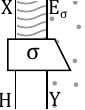}
~~=~~
\includegraphics[valign=c]{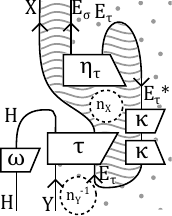}
\end{align}
Here $\eta_{\tau}: \id_{[m]} \to E_{\sigma} \otimes E_{\tau}$ is some isometry. We discussed how the equation on the LHS of the implication precisely corresponds to  entanglement-reversibility at the beginning of the proof of Proposition~\ref{prop:qbijmindil}. The implication follows by bending the $E_{\tau}$-wire down, precomposing by $\mathbbm{1}_X \otimes \kappa^{\dagger}$, then precomposing with the dagger of the rightmost 2-morphism of~\eqref{eq:biisomeqs} and using unitarity of that 2-morphism. Now, since two dilations related by an isometry on the environment are equivalent, and the 2-morphism on the RHS of~\eqref{eq:sigmaimplication} differs from~\eqref{eq:entinvgenmindil} only by a 2-morphism on the environment, we need only show that that 2-morphism is an isometry, which is seen as follows:
\begin{align*}
\includegraphics[valign=c]{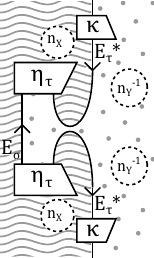}
~~=~~
\includegraphics[valign=c]{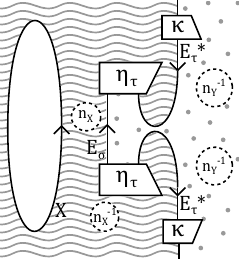}
~~=~~
\includegraphics[valign=c]{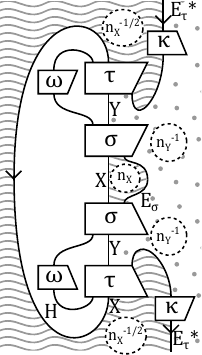}
~~=~~
\includegraphics[valign=c]{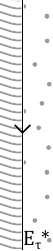}
\end{align*}
Here for the first equality we used $d_X = n_X^2$; for the second equality we used the first equation of~\eqref{eq:sigmaimplication} and its dagger; and for the third equality we used trace preservation~\eqref{eq:tracepresisom} for $\sigma$ and the fact that the rightmost morphism of~\eqref{eq:biisomeqs} is an isometry.

\vspace{.2cm}
\noindent
\emph{Proof of 3.} For the convenience of the reader we restate the claim in (3):
\begin{quote}
    The channel $(M,H)$ is entanglement-invertible with respect to $W$ precisely when the following conditions are satisfied:
\begin{itemize}
\item $(M,H)$ is a quantum bijection.
\item The linear map $\omega^{\dagger} \circ \omega :H \to H$ is an intertwiner $(M,H) \to (M,H)$.
\end{itemize}
\end{quote}
For the first direction, let us suppose that $(M,H)$ is a quantum bijection and that $\omega^{\dagger} \circ \omega$ is an intertwiner $(M,H) \to (M,H)$; we will then show that $M$ is entanglement-invertible with respect to $W$. 

First observe the following:
\begin{align}\label{eq:proof3lem1}
\includegraphics[valign=c]{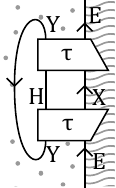}
~~=~~
\includegraphics[valign=c]{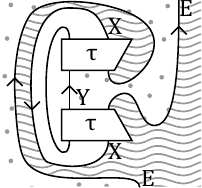}
~~=~~
\includegraphics[valign=c]{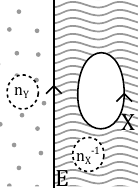}
\end{align}
Here the first equality is by isotopy (pull $\tau^{\dagger}$ around the $Y$-loop); the second equality is by the fact that the second morphism of~\eqref{eq:qbijbiueqs} is an isometry. 

Now consider the following positive element $x \in \End(E)$:
\begin{align*}
\includegraphics[valign=c]{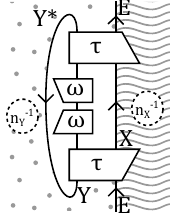}
\end{align*}
This is invertible by invertibility of $\omega$, minimality of the dilation $\tau$, and the intertwiner condition~\eqref{eq:intertwinerdef}. To see this, compose with the same element but with $\omega^{-1}$ and use the intertwiner condition to get rid of the $\omega$'s; then use invertibility of~\eqref{eq:mindil2morph}. We define $\kappa^T \in \End(E)$ to be the inverse of the positive square root of $x$, so that $x=(\kappa^{-1})^T (\kappa^{-1})^*$.

Now we have the following equation:
\begin{align}\label{eq:kappapullthrough}
\includegraphics[valign=c]{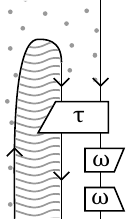}
~~=~~
\includegraphics[valign=c]{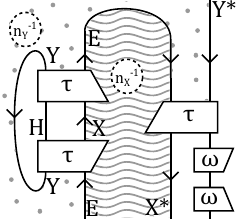}
~~=~~
\includegraphics[valign=c]{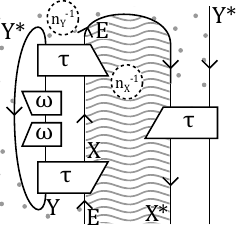}
~~=~~
\includegraphics[valign=c]{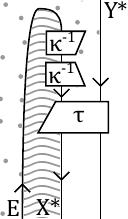}
\end{align}
Here the first equality is by~\eqref{eq:proof3lem1}; the second equality is by the fact that $\omega^{\dagger} \omega$ is an intertwiner~\eqref{eq:intertwinerdef}; and the third equality is by definition of $\kappa$. 

We can now show that $(M,H)$ is entanglement-invertible w.r.t. $W$. We will begin by showing entanglement-reversibility, which by Part 1 corresponds to showing that the second 2-morphism of~\eqref{eq:biisomeqs} is an isometry:
\begin{align*}
\includegraphics[valign=c]{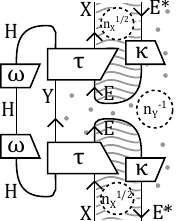}
~~=~~
\includegraphics[valign=c]{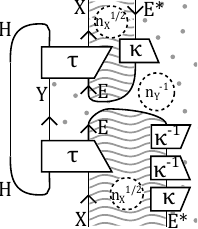}
~~=~~
\includegraphics[valign=c]{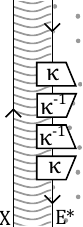}
~~=~~
\includegraphics[valign=c]{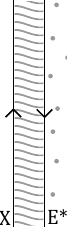}
\end{align*}
Here the first equality is by~\eqref{eq:kappapullthrough}; the second equality is by unitarity of the second 2-morphism of~\eqref{eq:qbijbiueqs}; and the third equality is by positivity of $\kappa$.

Now we know that $\dim(A) = \dim(B)$, so by Part 2 it follows that the 2-morphisms~\eqref{eq:biisomeqs} are unitary and the entanglement-left inverse has minimal dilation~\eqref{eq:entinvgenmindil}.

For this to be an entanglement-inverse we need to show the second equation of~\eqref{eq:entreventinv}. In terms of the dilations this is seen as follows:
\begin{align*}
\includegraphics[valign=c]{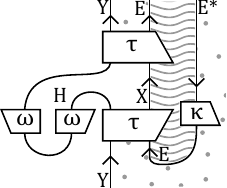}
~~=~~
\ignore{
\includegraphics[valign=c]{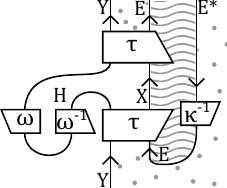}
~~=~~}
\includegraphics[valign=c]{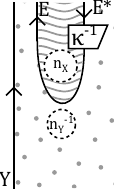}
\end{align*}
Here the equality is by~\eqref{eq:kappapullthrough} and the fact that the first 2-morphism of~\eqref{eq:qbijbiueqs} is a coisometry.
We have therefore shown that $(M,H)$ is entanglement-invertible w.r.t. $W$.

In the other direction, suppose that $(M,H)$ is entanglement-invertible w.r.t. $W$; we will show that $(M,H)$ is a quantum bijection and that $\omega^{\dagger} \omega$ is an intertwiner. Since there is an entanglement-reversible channel $A \to B$ and also an entanglement-reversible channel $B \to A$, we must have $\dim(A) = \dim(B)$ by Part 2. We therefore also know  from Part 2 that the minimal dilation $\sigma$ of the entanglement-inverse has the following form:
\begin{align}\label{eq:inproofentinvdil}
\includegraphics[valign=c]{pictures/entinvchans/entrevdef.pdf}
\end{align} 
In terms of the dilations, the second equation of~\eqref{eq:entreventinv} is as follows:
\begin{align}\label{eq:entinvkappa}
\includegraphics[valign=c]{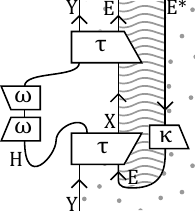}
~~=~~
\includegraphics[valign=c]{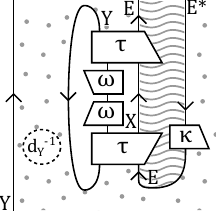}
~~=~~
\includegraphics[valign=c]{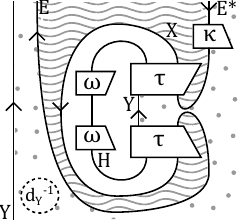}
~~=~~
\includegraphics[valign=c]{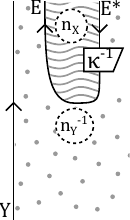}
\end{align}
Here the first equality is by the standard argument we made at the beginning of the proof of Proposition~\ref{prop:qbijmindil} (note that we computed the isometry $\eta_{\sigma}: \id_{[n]} \to E \otimes E^*$ by tracing out the $Y$-wire); the second equality is by a topological manipulation, pulling $\tau^{\dagger}$ around the $Y$-wire; and the third equality is by the fact that the second morphism of~\eqref{eq:biisomeqs} is an isometry.

The entanglement-inverse channel with minimal dilation~\eqref{eq:inproofentinvdil} is itself entanglement-invertible, and by Part 2 the minimal dilation of its entanglement-inverse can be obtained by inserting~\eqref{eq:inproofentinvdil} into~\eqref{eq:entinvgenmindil}. But we know that the entanglement-inverse of the entanglement-inverse is the original channel, so by uniqueness of the minimal dilation up to a unitary on the environment we obtain the following equation for $\tau$, where $\kappa' \in \End(E^*)$ is some isomorphism: 
\begin{align}\label{eq:kappadouble}
\includegraphics[valign=c]{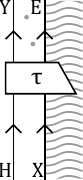}
~~=~~
\includegraphics[valign=c]{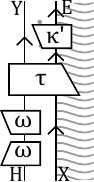}
\end{align}
Precomposing~\eqref{eq:kappadouble} by $\tau^{\dagger}$, and using~\eqref{eq:entinvkappa} and unitarity of the first 2-morphism of~\eqref{eq:biisomeqs}, we obtain $\kappa' = \kappa^{\dagger} \kappa$. Now we observe that the rightmost 2-morphism of~\eqref{eq:qbijbiueqs} is an isometry:
\begin{align*}
\includegraphics[valign=c]{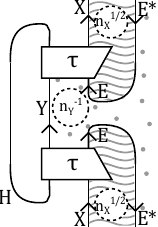}
~~=~~
\includegraphics[valign=c]{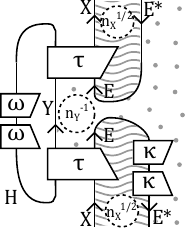}
~~=~~
\includegraphics[valign=c]{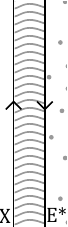}
\end{align*}
Here the first equality is by~\eqref{eq:kappadouble}, and the second equality is by the fact that the second morphism of~\eqref{eq:biisomeqs} is an isometry. Since $\dim(A) = \dim(B)$ we have by Part 2 that the morphisms~\eqref{eq:qbijbiueqs} are furthermore unitary, and therefore $(M,H)$ is a quantum bijection by Proposition~\ref{prop:qbijmindil}. Finally, we see by~\eqref{eq:kappadouble} that $\omega^{\dagger} \circ \omega$ is an intertwiner $(M,H) \to (M,H)$:
\begin{align*}
\includegraphics[valign=c]{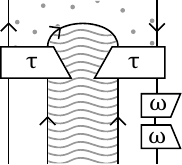}
~~=~~
\includegraphics[valign=c]{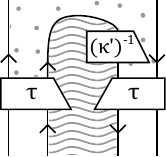}
~~=~~
\includegraphics[valign=c]{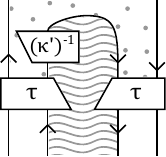}
~~=~~
\includegraphics[valign=c]{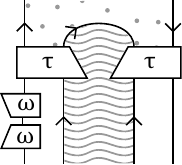}
\end{align*}
Here the first and third equalities are by~\eqref{eq:kappadouble}, and the second equality is by isotopy of the diagram. 
\end{proof}

\paragraph{Result for mixed states.} We now generalise the result to mixed states. First, some definitions. Let $M: A \otimes B(H_1) \to B$ be a channel, and let $\tau: H_1 \otimes X \to Y \otimes E$ be a minimal dilation. Let $W: \mathbb{C} \to B(H_1) \otimes B(H_2)$ be a state. Now $W$ is a convex combination of pure states $\{W_i\}_{i \in I}$, each of which is defined by some $\omega_i: H_1 \to H_2$; as before, let $\{\iota_i, q_i\}_{i \in I}$ be the isometries and coisometries such that $\omega_i = \iota_i \circ \bar{\omega}_i \circ q_i$, with $\bar{\omega}_i$ invertible, and let $\bar{W_i}: \mathbb{C} \to  B(H_1/\Ker(\omega_i)) \otimes B(\Image(\omega_i))$ be the pure states defined by $\bar{\omega}_i$. For each $i \in I$, let $\bar{M}_i: A \otimes B(H_1 / \Ker(\omega_i)) \to B$ be the channel whose dilation is a scalar multiple of $\tau \circ (q_i^T \otimes \mathbbm{1}_X)$ (where the scalar multiplier is chosen such that the trace preservation condition~\eqref{eq:tracepresisom} is satisfied).
\begin{corollary}\label{cor:genmixed}
Using the definitions and notation from the previous paragraph:
\begin{enumerate}
\item The channel $(M,H_1)$ is entanglement-reversible with respect to $W$ precisely when there exist some $2$-morphisms $\nu_{ij}: E^* \to E^*$ such that, for all $i,j \in I$:
\begin{align}\label{eq:biisomixed}
\includegraphics[valign=c]{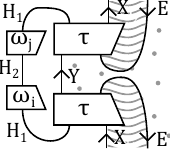}
~~=~~
\includegraphics[valign=c]{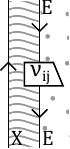}
\end{align}
In this case, $\dim(A) \leq \dim(B)$.
\item The channel $(M,H_1)$ is entanglement-invertible with respect to $W$ precisely when the following conditions hold:
\begin{itemize}
\item $(M,H_1)$ is entanglement-reversible with respect to $W$.
\item Each of the channels $\bar{M}_i$ is entanglement-invertible with respect to the state $\bar{W}_i$.
\end{itemize}
\end{enumerate}
\end{corollary}
\begin{proof}
We prove the statements in order.

\vspace{.2cm}
\noindent
\emph{Proof of 1.} 
The proof is similar to the proof of Part 1 of Theorem~\ref{thm:genpure}. The confusability graph of~\eqref{eq:chantoreverse} is now as follows:
\begin{align*}
\textrm{supp}~\left(\sum_{i,j}
\includegraphics[valign=c]{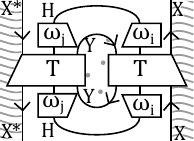}\right)
\end{align*}
The same argument as in the proof of Theorem~\ref{thm:genpure} then shows that reversibility of~\eqref{eq:chantoreverse} is equivalent to the existence of some $2$-morphisms $\nu_{ij}: E^* \to E^*$ such that the equation~\eqref{eq:biisomixed} is obeyed.

\vspace{.2cm}
\noindent
\emph{Proof of 2.}  Clearly, $(M,H_1): A \to B$ is entanglement-invertible w.r.t. $W$ precisely when there exists a channel $(N,H_2): B \to A$ which is an entanglement-inverse for $(M,H_1)$ w.r.t. all the $W_i$. 

We show that the stated conditions imply this. Since $(M,H_1)$ is entanglement-reversible w.r.t. $W$, there is a channel $(N,H_2): B \to A$ which is an entanglement-left inverse for $(M,H_1)$ w.r.t. all the $W_i$. Let the minimal dilation of $N$ be $\sigma: H_2 \otimes Y \to X \otimes E$. The entanglement-reversibility equation gives the following equations for the dilations:
\begin{align}\label{eq:mixedentrevdil}
\includegraphics[valign=c]{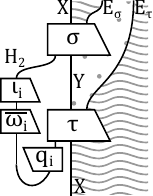}
~~=~~
\includegraphics[valign=c]{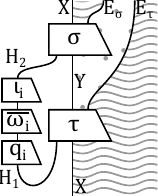}
~~=~~
\includegraphics[valign=c]{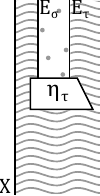}
\end{align}
Let $\bar{N}_i: B \otimes B(\Image(\omega_i)) \to A$ be the channel whose dilation is a scalar multiple of $\sigma \circ (\iota_i \otimes \mathbbm{1}_Y)$. The equation~\eqref{eq:mixedentrevdil} tells us that $\bar{N}_i$ is an entanglement-left inverse for $\bar{M}_i$ w.r.t. $\bar{W}_i$. Since $\dim(A) = \dim(B)$, this entanglement-left inverse channel is uniquely defined by Part 2 of Theorem~\ref{thm:genpure}. Since we have assumed that $\bar{M}_i$ is entanglement-invertible w.r.t $\bar{W}_i$, $\bar{N}_i$ must be the entanglement-inverse for $\bar{M}_i$ w.r.t. $\bar{W}_i$; thus the following equation is also satisfied:
\begin{align*}
\includegraphics[valign=c]{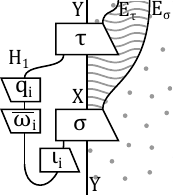}
~~=~~
\includegraphics[valign=c]{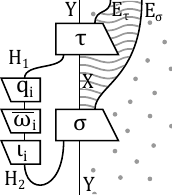}
~~=~~
\includegraphics[valign=c]{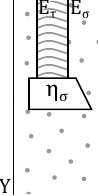}
\end{align*}
But this equation (looking at the second and third terms) says precisely that $(N,H_2)$ is an entanglement-inverse for $M$ w.r.t. $W_i$.

In the other direction, suppose that there exists a channel $(N,H_2)$ which is an entanglement-inverse for $(M,H_1)$ w.r.t. all the $W_i$. Then $(M,H_1)$ is entanglement-reversible w.r.t. $W$ by definition. The other conditions follow by Lemma~\ref{lem:entrevquotient}.
\end{proof}
\begin{remark}
The reader will observe that we made no statement about uniqueness of the entanglement-left inverse in Corollary~\ref{cor:genmixed}. This is because we do not have uniqueness even for pure $W$ when $\omega$ is not invertible, since the extension in the proof of Lemma~\ref{lem:entrevquotient} is non-unique.
\end{remark}

\subsubsection{Example: Werner's classification of tight teleportation and dense coding schemes}\label{sec:wernerex}

Finally, it may be useful, particularly for readers unfamiliar with the graphical techniques used in this work, to see how Theorem~\ref{thm:genpure} implies Werner's classification of tight teleportation and dense coding schemes in terms of unitary error bases~\cite{Werner2001}.

\begin{definition}
Let $H$ be a Hilbert space of dimension $d$. 
\begin{itemize}
\item A \emph{tight teleportation scheme} is a pair $(W,M)$ of a state $W: \mathbb{C} \to B(H) \otimes B(H)$ and a channel $(M,H): B(H) \to{} [d^2]$ which is entanglement-reversible with respect to $W$.
\item A \emph{tight dense coding scheme} is a pair $(W,N)$ of a state $W: \mathbb{C} \to B(H) \otimes B(H)$ and a channel  $(N,H):  [d^2] \to B(H)$ which is entanglement-reversible with respect to $W$.
\end{itemize}
\end{definition}
\begin{example}[Unitary error bases]\label{ex:uebchannels}
A \emph{unitary error basis} for a Hilbert space $H$ of dimension $d$ is a basis $\{U_i\}_{i \in [d^2]}$ of unitary operators on $H$ orthogonal under the trace inner product, i.e. $\frac{1}{d}\Tr(U_j^{\dagger} U_i) = \delta_{ij}$. From a unitary error basis $\{U_i\}_{i \in [d^2]}$, we construct two channels.
\begin{itemize}
\item The channel $M: B(H) \otimes B(H) \cong B(H \otimes H) \to{} [d^2]$ is defined by a complete projective measurement in the orthonormal basis $\{\frac{1}{\sqrt{\dim(H)}} (U_i \otimes \mathbbm{1})\ket{\eta_H}\}_{i \in [d^2]}$ of $H \otimes H$.
\item The channel $N: [d^2] \otimes B(H) \to B(H)$ is a controlled unitary operation, where the classical control $i \in [d^2]$ corresponds to the unitary $U_i^{\dagger}$.
\end{itemize}
These channels are quantum bijections. The channel $(M,H)$ therefore specifies a tight teleportation scheme, and the channel $(N,H)$ a tight dense coding scheme. Moreover, $(N,H)$ is the (unique, by Theorem~\ref{thm:genpure}) entanglement-inverse of $(M,H)$. 
\end{example}
\noindent
We will prove the following result as a corollary of Theorem~\ref{thm:genpure}.
\begin{corollary}[{\cite[Thm. 1]{Werner2001}}]
The following statements hold:
\begin{itemize}
\item Let $(W,M)$ be a tight teleportation scheme. Then $W$ is a maximally entangled pure state and $(M,H): B(H) \to{} [d^2]$ is a quantum bijection defined by a unitary error basis as in Example~\ref{ex:uebchannels}.
\item Let $(W,N)$ be a tight dense coding scheme. Then $W$ is a maximally entangled pure state and $(N,H): [d^2] \to B(H)$ is a quantum bijection defined by a unitary error basis as in Example~\ref{ex:uebchannels}.
\end{itemize}
This yields a bijection between tight dense coding schemes and tight teleportation schemes. 
\end{corollary}
\begin{proof} 
We observe that if we can prove the two bullet pointed statements, the final statement follows immediately, since the entanglement-inverse of a tight teleportation scheme is a tight dense coding scheme, and vice versa.  We prove the bullet pointed statements as follows.

\vspace{.2cm}
\noindent
\emph{Tight teleportation schemes.} Let us assume that the state $W$ is pure; we will remove this assumption at the end. We furthermore assume that the linear map $\omega: H \to H$ defining $W$ is invertible; we will remove this assumption at the end.

We are therefore in the situation of Theorem~\ref{thm:genpure}. We split the algebras as $B(H) \cong H \otimes H$ and $[d^2] \cong X \otimes X^*$, where $X: [1] \to{} [d^2]$ is defined by $X := (\mathbb{C})_{i \in [d^2]}$. Let $\tau: H \otimes X \to Y \otimes E$ be a minimal dilation of $M$. The 2-morphisms~\eqref{eq:biisomeqs} are unitaries. Choosing an index $i \in [d^2]$ for the nontrivial region, we observe that unitarity of the first 2-morphism of~\eqref{eq:biisomeqs} implies the following equality:
\begin{align}\label{eq:tightteltrace}
d^2 = 
\sum_{i \in [d^2]}~
\includegraphics[valign=c]{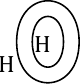}
~~=
\sum_{i \in [d^2]}~
\includegraphics[valign=c]{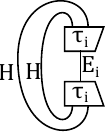}
~~=
\sum_{i \in [d^2]}~
\includegraphics[valign=c]{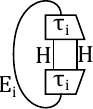}
~~=
\sum_{i \in [d^2]}~
\includegraphics[valign=c]{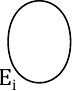}
~~=\sum_{i \in [d^2]} \dim(E_i)
\end{align}
Here the first equality is by~\eqref{eq:tracedef}; the second equality is by the fact that the first 2-morphism of~\eqref{eq:biisomeqs} is an isometry; the third equality is by isotopy, pulling $\tau_i$ around the loop; the fourth equality is by the fact that the first 2-morphism of~\eqref{eq:biisomeqs} is a coisometry; and the final equality is by~\eqref{eq:tracedef}. It follows that $E_{i} \cong \mathbb{C}$ for all $i \in [d^2]$.

It follows that $\tau_i: H \otimes H \to \mathbb{C}$, and $\kappa_i$ are some scalars. Let $u_i \in \End(H)$ be such that $\tau_i = \bra{\eta_H} (\mathbbm{1} \otimes u_i^{\dagger})$. Then unitarity of the 2-morphisms~\eqref{eq:biisomeqs} reduces to the following equations for $\{u_i\}_{i \in [d^2]}$:
\begin{align}\label{eq:biisom1ueb}
\frac{1}{d}\sum_{i\in [d^2]}~
\includegraphics[valign=c]{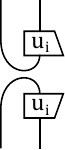}
~~=~~
\includegraphics[valign=c]{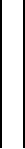}
&&
\frac{1}{d}~
\includegraphics[valign=c]{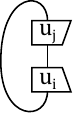}
~~=~~
\delta_{ij}
\\\label{eq:biisom2ueb}
|\kappa_i|^2~
 \includegraphics[valign=c]{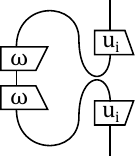}
~~=~~
\includegraphics[valign=c]{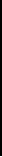}
&&
|\kappa_i|^2~
 \includegraphics[valign=c]{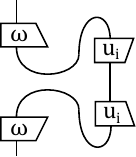}
~~=~~
\includegraphics[valign=c]{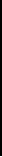}
\end{align}
Here from left to right the equations~\eqref{eq:biisom1ueb} correspond to isometry and coisometry of the first 2-morphism of~\eqref{eq:biisomeqs}, while the equations~\eqref{eq:biisom2ueb} correspond to isometry and coisometry of the second 2-morphism of~\eqref{eq:biisomeqs}. 

The equations~\eqref{eq:biisom1ueb} state precisely that the $\{U_i\}_{i \in [d^2]}$ form an orthonormal basis of $B(H)$ under the Hilbert-Schmidt inner product. By invertibility of $\omega$, the second equation of~\eqref{eq:biisom2ueb} implies that $U_i^{\dagger} U_i = \frac{1}{|\kappa_i|^2} \omega^{-1} (\omega^{\dagger})^{-1}$. The second equation of~\eqref{eq:biisom1ueb} then implies that $\frac{ \Tr(\omega^{-1} (\omega^{\dagger})^{-1})}{d|\kappa_i|^2} = \frac{\Tr(U_i^{\dagger} U_i)}{d} = 1$, so in particular $|\kappa_i|^2 =: |\kappa|^2$ is a constant.

From the second equation of~\eqref{eq:biisom2ueb} we then make the following deduction:
\begin{align*}
\includegraphics[valign=c]{pictures/wernerex/unitarity42.pdf}
~~=~~
\frac{|\kappa|^2}{d^2}
\sum_{i=1}^{d^2}
\includegraphics[valign=c]{pictures/wernerex/unitarity41.pdf}
~~=~~
\frac{|\kappa|^2}{d^2}
\sum_{i=1}^{d^2}
\includegraphics[valign=c]{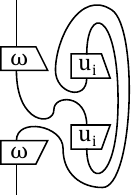}
~~=~~
|\kappa|^2
\includegraphics[valign=c]{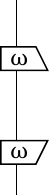}
\end{align*}
Here the first equality is by the second equation of~\eqref{eq:biisom2ueb}; the second equality is isotopy of the diagram; and the third equality is by the first equation of~\eqref{eq:biisom1ueb}. We see that $|\kappa|^2 \omega \circ \omega^{\dagger} = \mathbbm{1}$, and so $|\kappa| \omega$ is an coisometry and therefore unitary. \ignore{Since $W$ dilates a state, we have $1 = \dim(H) \Tr(\omega^{\dagger} \circ \omega) = \frac{\dim(H)^2}{\dim(H)|\kappa|^2}$, which implies $|\kappa| = \sqrt{\dim(H)}$.} The state $W$ is therefore maximally entangled. 

Unitarity of the 2-morphisms~\eqref{eq:qbijbiueqs} follows immediately from unitarity of the 2-morphisms~\eqref{eq:biisomeqs} and unitarity of $|\kappa|\omega$. Therefore $(M,H)$ is a quantum bijection by Proposition~\ref{prop:qbijmindil}. 

We now remove the assumption that the linear map $\omega$ defining $W$ is invertible. Suppose that it is not invertible; then by Lemma~\ref{lem:entrevquotient}, $(\bar{M}, H/\Ker(\omega))$ must be the channel for a tight teleportation scheme with the state $\bar{W}$. Then, from what we have shown already, there is a unitary $H_1/\Ker(\omega) \otimes H \to X \otimes E$, where $E$ is some environment. But as we saw in~\eqref{eq:tightteltrace}, this implies that $\dim(H/\Ker(\omega))\dim(H) = \sum_{i \in [d^2]}\dim(E_i)$. Since $\dim(E_i) \geq 1$ it follows that $\dim(H/\Ker(\omega)) = \dim(H)$, and so $\omega$ is invertible. 

Finally, we remove the assumption that $W$ is pure. Let $W: \mathbb{C} \to B(H) \otimes B(H)$ be a mixed state, which is a convex combination of pure states $W_i: \mathbb{C} \to B(H) \otimes B(H)$, where $i \in I$. An entanglement-left inverse $(N,H)$ for $(M,H)$ w.r.t. $W$ must be an entanglement-left inverse w.r.t. all the pure states $W_i$ independently. Therefore, by what has already been said, the $W_i$ are all maximally entangled, defined by unitaries $\omega_i: H \to H$; moreover, the channel $M$ is defined by a unitary error basis $\{u_k\}_{k \in [d^2]}$. The equation~\eqref{eq:biisomixed} requires that, for each $i,j \in I$ and $k \in [d^2]$:
\begin{align*}
\includegraphics[valign=c]{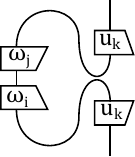}
~~=~~
\kappa_j \kappa_i^{\dagger}~ \includegraphics[valign=c]{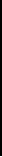}
\end{align*}
But by unitarity of the $\{u_k\}$ this implies that $\omega_j^{\dagger} \omega_i$ is proportional to the identity, which implies $\omega_i = \omega_j$ by unitarity of the $\{\omega_i\}$. We therefore have $W_i = W_j$ for all $i,j \in I$, and the state $W$ is pure.

\vspace{.2cm}
\noindent
\emph{Tight dense coding schemes.}
Again, to begin with we assume that the state $W$ is pure and defined by an invertible linear map $\omega: H \to H$.
We are therefore in the situation of Theorem~\ref{thm:genpure}.

We again split the algebras as $B(H) \cong H \otimes H$ and $[d^2] \cong X \otimes X^*$, where $X: [1] \to{} [d^2]$ is defined by $X = (\mathbb{C})_{i \in [d^2]}$. Let $\tau: H \otimes X \to H \otimes E$ be a minimal dilation of $N$. By a similar argument to~\eqref{eq:tightteltrace}, unitarity of the first 2-morphism of~\eqref{eq:biisomeqs} implies that $d^3 = d \sum_{i=1}^{d^2} \dim(E_i)$, so $E_i \cong \mathbb{C}$ for all $i$. Then $\tau_i: H \to  H$, and $\kappa_i$ are scalars. Unitarity of the morphisms~\eqref{eq:biisomeqs} implies the following equations for $\tau_i$:
\begin{align}\label{eq:dcbiueqs1}
\tau_i^{\dagger} \circ \tau_i = \mathbbm{1} && \tau_i \circ \tau_i^{\dagger} = \mathbbm{1} 
\\\label{eq:dcbiueqs2}
 \frac{\kappa_i \kappa_j^*}{d}~\includegraphics[valign=c]{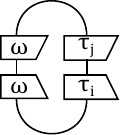}~~=~~\delta_{ij}
&&  \sum_{i \in [d^2]} \frac{|\kappa_i|^2}{d}~\includegraphics[valign=c]{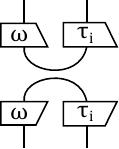}
 ~~=~~ \includegraphics[valign=c]{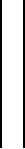}
\end{align}
The equations~\eqref{eq:dcbiueqs1} say precisely that the $\{\tau_i\}$ are unitary. Setting $i = j$ in the left hand equation of~\eqref{eq:dcbiueqs2} we obtain $1= \frac{|\kappa_i|^2}{d}\Tr(\omega^{\dagger} \omega) = \frac{|\kappa_i|^2}{d^2}$, which implies that $|\kappa_i|^2 = d^2$. (Here we calculated $\Tr(\omega^{\dagger} \omega)=\frac{1}{d}$ using the fact that $(\omega \otimes \mathbbm{1}) \ket{\eta_H}$ is a minimal dilation.) Now, taking the trace of the rightmost wire in the second equation of~\eqref{eq:dcbiueqs2} and using unitarity of $\tau_i$, we obtain $d^2 \omega \omega^{\dagger} = \mathbbm{1}$, which implies that $d\omega$ is unitary; $W$ is therefore a maximally entangled state. It then follows from the first equation of~\eqref{eq:dcbiueqs2} that $\frac{1}{d}\Tr(\tau_j^{\dagger} \tau_i) = \delta_{ij}$, i.e. that the unitaries $\{\tau_i\}_{i \in [d^2]}$ are orthogonal under the Hilbert-Schmidt inner product and therefore form a unitary error basis. 

Unitarity of the 2-morphisms~\eqref{eq:qbijbiueqs} follows immediately from unitarity of the 2-morphisms~\eqref{eq:biisomeqs} and unitarity of $d\omega$; therefore, by Proposition~\ref{prop:qbijmindil}, $(N,H)$ is a quantum bijection. 

We can remove the assumption that $\omega$ is invertible using a similar argument to that made above for tight teleportation schemes.

Finally, we remove the assumption that $W$ is pure. Let $W: \mathbb{C} \to B(H) \otimes B(H)$ be a convex combination of pure states $W_i: \mathbb{C} \to B(H) \otimes B(H)$, where $i \in I$. An entanglement-left inverse $(N,H)$ for $(M,H)$ w.r.t. $W$ must be an entanglement-left inverse w.r.t. all the states $W_i$ independently. Therefore, by what has already been said, the $W_i$ are all maximally entangled, defined by unitaries $\omega_i: H \to H$; moreover, the channel $M$ is defined by a unitary error basis $\{u_k\}_{k \in [d^2]}$. The equation~\eqref{eq:biisomixed} requires that, for each $i,j \in I$ and $k,l \in [d^2]$, there exists some scalar $c_{ij}$ such that:	
\begin{align*}
\includegraphics[valign=c]{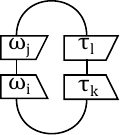}
~~=~~
c_{ij} \delta_{kl} 
\end{align*}
By orthonormality of $\tau_{k}$ this implies that $\omega_j^{\dagger} \omega_i$ is a scalar multiple of the identity for each $i,j \in I$; which, by unitarity of the $\{\omega_i\}$, implies that they are identical. $W$ is therefore a pure state.
\end{proof}

\bibliographystyle{alphaurl}
\bibliography{bibliography}

\end{document}